\documentclass{article}
\pdfoutput=1
\usepackage{geometry}
\usepackage[T1]{fontenc}
\usepackage[utf8]{inputenc}

\usepackage{amsmath,amssymb}
\usepackage{amsthm}

\usepackage{axodraw2}
\usepackage{subcaption}
\usepackage{tikz}
\usetikzlibrary{positioning,patterns}

\usepackage[pdfauthor={Robert Beekveldt, Michael Borinsky, Franz Herzog},
            pdftitle={The Hopf algebra structure of the R*-operation},
            pdfsubject={The Hopf algebra structure of the R*-operation},
            pdfkeywords={Feynman integrals, renormalisation, Hopf algebra, infrared singularities, IR rearrangement, UV-IR duality}]{hyperref}

\newtheorem{theorem}{Theorem}
\newtheorem{corollary}{Corollary}
\newtheorem{lemma}{Lemma}
\newtheorem{assumption}{Assumption}

\def\beq{\begin{equation}}   
\def\eeq{\end{equation}}
\def\bea{\begin{eqnarray}}  
\def\eea{\end{eqnarray}} 
\def\nn{\nonumber}
\def\r{\right} 
\def\l{\left} 
\def\f21{{}_2F_{1}}
\def\eps{\epsilon}

\def\M{\mathcal{M}}
\def\H{\mathcal{H}}

\def\O{\mathcal{O}}
\def\A{\mathcal{A}}

\def\IR{\text{IR}}

\def\t{\widetilde}

\def\Q{\mathbb{Q}}
\DeclareMathOperator{\id}{id}
\def\one{\mathbb{I}}
\def\counit{\epsilon}
\def\unit{u}
\def\to{\rightarrow}
\def\UV{\mathrm{UV}}

\def\ZZ{Z}

\def\bubbleonenum#1#2{
\raisebox{-7pt}
{
\begin{axopicture}{(30,20)(-13,-10)}
\SetScale{1}\SetColor{black}%
\Line[width=0.25](-15,0)(-10,0)
\Line[width=0.25](10,0)(15,0) 
\Arc[width=0.25](0,0)(10,0,360)
\Vertex(-10,0){1.5}
\Vertex(10,0){1.5}
\SetScale{1}\SetColor{black}%
\Text(1.5,6.5){\tiny $#1$ \tiny}
\Text(1.5,-6.5){\tiny $#2$ \tiny}
\end{axopicture}
}
}

\def\bubbletwonum#1#2#3#4{
\raisebox{-7pt}
{
\begin{axopicture}{(30,20)(-13,-10)}
\SetScale{1}%
\Line[width=0.25](-15,0)(-10,0)
\Line[width=0.25](10,0)(15,0) 
\Arc[width=0.25](0,0)(10,0,360)
\Arc[width=0.25](10,10)(10,180,270)
\Vertex(-10,0){1.5}
\Vertex(10,0){1.5}
\Vertex(0,10){1.5}
\SetScale{1}%
\Text(-5,4){\tiny $#1$ \tiny}
\Text(2,1.5){\tiny $#2$ \tiny}
\Text(6,5.5){\tiny $#3$ \tiny}
\Text(0,-7){\tiny $#4$ \tiny}
\end{axopicture}
}
}

\def\bubbletwodotnumc#1#2#3#4{
\raisebox{-7pt}
{
\begin{axopicture}{(30,20)(-13,-10)}
\SetScale{1}\SetColor{black}%
\Line[width=0.25](-15,0)(-10,0)
\Line[width=0.25](10,0)(15,0) 
\Line[width=0.25](-10,0)(10,0) 
\Arc[width=0.25](0,0)(10,0,360)
\Vertex(0,0){1.5}
\Vertex(-10,0){1.5}
\Vertex(10,0){1.5}
\SetScale{1}\SetColor{black}%
\Text(1.5,13){\tiny $#1$ \tiny}
\Text(-3,3){\tiny $#2$ \tiny}
\Text(5.5,3){\tiny $#3$ \tiny}
\Text(1.5,-6.5){\tiny $#4$ \tiny}
\end{axopicture}
}
}

\def\tadpolenumrem#1#2{
\raisebox{-10pt}
{
\begin{axopicture}{(15,20)(-13,-10)}
\SetScale{1}\SetColor{black}%
\Line[width=0.25](-10,0)(0,0)
\Arc[width=0.25](-5,5)(5,0,360)
\Vertex(-5,0){1.5}
\Vertex(-5,10){1.5}
\SetScale{1}\SetColor{black}%
\Text(-6,4.5){\tiny $#1$ \tiny}
\Text(-1.5,4.5){\tiny $#2$ \tiny}
\end{axopicture}
}
}

\def\vaconenum#1#2{
\raisebox{-7pt}
{
\begin{axopicture}{(25,20)(-11,-10)}
\SetScale{1}\SetColor{black}%
\Arc[width=0.25](0,0)(10,0,360)
\Vertex(-10,0){1.5}
\Vertex(10,0){1.5}
\SetScale{1}\SetColor{black}%
\Text(1.5,6.5){\tiny $#1$ \tiny}
\Text(1.5,-6.5){\tiny $#2$ \tiny}
\end{axopicture}
}
}

\def\contractgraph1{
  \begin{axopicture}(530,60) (0,100)
    \SetWidth{1.0}
    \SetColor{black}
    \Arc[clock](288,62)(200,163.74,16.26)
    \Arc[width=0.25](288,174)(200,-163.74,-16.26)
    \Line[width=0.25](16,230)(96,118)
    \Line[width=0.25](480,118)(560,6)
    \Line[width=0.25](560,230)(480,118)
    \Line[width=0.25](96,118)(16,6)
    \Vertex(96,118){12}
    \Vertex(480,118){12}
    \Line[width=0.25](96,118)(480,118)
    \Vertex(288,260){12}
    \Vertex(288,-24){12}
  \end{axopicture}
}

\def\graphcut2{
  \begin{axopicture}(170,50) (270,70)
    \SetWidth{1.0}
    \SetColor{black}

	\ECirc(360,100){40}
	\Vertex(360,60){10}
	\ECirc(360,180){40}
	\Vertex(360,140){10}
	\ECirc(360,20){40}

\end{axopicture}
}

\def\contract77{
  \begin{axopicture}(150,50) (270,50)
    \SetWidth{1.0}
    \SetColor{black}

	\ECirc(360,100){40}
	\Vertex(360,60){10}
	\Vertex(360,140){10}
	\ECirc(360,20){40}

\end{axopicture}
}

\def\graphR{
\begin{axopicture}(150,60) (50,0)
    \SetWidth{1.0}
    \SetColor{black}
    \Arc[width=0.25](128,12)(45,135,495)
    \Line[width=0.25](80,12)(176,12)
    \Vertex(83,12){6}
    \Vertex(173,12){6}
    \Vertex(128,57){6}

  \end{axopicture}
}

\def\fullext{
\begin{axopicture}(194,92) (31,0)
    \SetWidth{1.0}
    \SetColor{black}
    \Arc[width=0.25](128,12)(45,135,495)
    \Line(32,12)(80,12)
    \Line(224,12)(176,12)
    \Line[width=0.25](80,12)(176,12)
    \Vertex(83,12){6}
    \Vertex(173,12){6}
    \Vertex(128,57){6}

     \fontsize{6}{100}\selectfont
    \Text(118,17)[lb]{2}
    \Text(118,-27)[lb]{3}
    \Text(80,50)[lb]{1}
    \Text(164,50)[lb]{4}
  \end{axopicture}
}

\def\subgraph23{
\begin{axopicture}(140,110)(55,0)
    \SetWidth{1.0}
    \SetColor{black}
    \Arc[width=0.25](128,12)(45,135,495)
    \Vertex(83,12){6}
    \Vertex(173,12){6}    

     \fontsize{6}{100}\selectfont
    \Text(123,-57)[lb]{2}
    \Text(123,65)[lb]{3}
\end{axopicture}
}

\def\extlegs23{
\begin{axopicture}(200,110)(31,0)
    \SetWidth{1.0}
    \SetColor{black}
    \Arc[width=0.25](128,12)(45,135,495)
    \Line(32,12)(80,12)
    \Line(224,12)(176,12)
    \Vertex(83,12){6}
    \Vertex(173,12){6}    

     \fontsize{6}{100}\selectfont
    \Text(123,-27)[lb]{2}
    \Text(123,65)[lb]{3}
\end{axopicture}
}
\def\externallegs14{
\begin{axopicture}(90,92) (0,40)
    \SetWidth{1.0}
    \SetColor{black}
	\Bezier(40,100)(10,55)(10,55)(40,10)
	\Bezier(40,100)(70,55)(70,55)(40,10)
	\Vertex(40,100){6}
	\Vertex(40,10){6}
	\Line(5,10)(40,10)
	\Line(80,10)(40,10)
	 \fontsize{6}{100}\selectfont
	\Text(-4,55)[lb]{1}
    \Text(71,55)[lb]{4}
\end{axopicture}
}

\def\subgraphR14irdiv{
\begin{axopicture}(90,120) (0,48)
    \SetWidth{1.0}
    \SetColor{black}
	\Bezier(40,100)(10,55)(10,55)(40,10)
	\Bezier(40,100)(70,55)(70,55)(40,10)
	\Vertex(40,100){6}
	\Vertex(40,10){6}

	 \fontsize{6}{100}\selectfont
	\Text(-4,55)[lb]{1}
    \Text(71,55)[lb]{4}
\end{axopicture}
}

\def\diagram2loopA{
  	\begin{axopicture}(80,45)(0,-5)
	\SetWidth{1.0}
    \SetColor{black}    
    \Line[width=0.25](0,0)(20,0)
    \Arc[width=0.25](40,0)(20,0,360)
    \Line[width=0.25](80,0)(60,0)
    \Vertex(20,0){3} %
    \Vertex(60,0){3} %
    \Vertex(40,20){3} %
    \Bezier(60,0)(45,6)(45,6)(40,20)
   \end{axopicture}
}

\def\adiagram1LoopA{
\begin{axopicture}(115,60)(0,0)
	\SetScale{0.35}
    \SetWidth{1.0}
    \SetColor{black}
    \Arc[width=0.25](128,12)(45,135,495)
    \Line(32,12)(80,12)
    \Line(224,12)(176,12)
    \Vertex(83,12){6}
    \Vertex(173,12){6}    

     \fontsize{6}{100}\selectfont
\end{axopicture}
}

\def\cograph1{
\begin{axopicture}(200,110)(31,0)
    \SetWidth{1.0}
    \SetColor{black}
    \Arc[width=0.25](128,12)(45,135,495)
    \Line(32,12)(80,12)
    \Line(224,12)(176,12)
    \Vertex(83,12){6}
    \Vertex(173,12){6}    
\end{axopicture}
}

\def\trianglefull{
\begin {axopicture}(46,30)(-15,-14)
    \SetColor{black}  
    \SetScale{0.55}      
    \Line[width=0.25](-20,10)(0,0)
    \Line[width=0.25](-20,-50)(0,-40)
    \Line[width=0.25](0,0)(20,-10)%
    \Line[width=0.25](20,-10)(40,-20)%
    \Line[width=0.25](0,-40)(20,-30)%
    \Line[width=0.25](20,-30)(40,-20)%
    \Line[width=0.25](20,-10)(20,-30)%
    \DoubleLine(0,-40)(0,0){3}%
    \Vertex(0,0){3}
    \Vertex(0,-40){3}
    \Vertex(20,-10){3}
    \Vertex(20,-30){3}
    \Vertex(40,-20){3}
\end{axopicture}
}

\def\1subgraph1{
\begin {axopicture}(46,30)(-15,-14)
    \SetColor{black}  
	\SetScale{0.55}  

    \Line[width=0.25](-20,10)(0,0)
    \Line[width=0.25](-20,-50)(0,-40)
    \Line[width=0.25](0,0)(20,-10)%
    \Line[width=0.25](20,-10)(40,-20)%

    \Line[width=0.25](0,-40)(20,-30)%
    \Line[width=0.25](20,-30)(40,-20)%

    \DoubleLine(0,-40)(0,0){3}%
    \Vertex(0,0){3}
    \Vertex(0,-40){3}
    \Vertex(20,-10){3}
    \Vertex(20,-30){3}
    \Vertex(40,-20){3}
   \end{axopicture}
}

\def\2subgraph2{
\begin {axopicture}(26,30)(-15,-14)
    \SetScale{0.55}  
	\Line[width=0.25](-20,10)(0,0)
    \Line[width=0.25](-20,-50)(0,-40)
    \DoubleLine(0,-40)(0,0){3}
    \Vertex(0,0){3}
    \Vertex(0,-40){3}
   \end{axopicture}
}

\def\3subgraph3{
\begin {axopicture}(36,30)(-15,-14)
	\SetScale{0.55}  	
	\Line[width=0.25](-20,10)(0,0)
    \Line[width=0.25](-20,-50)(0,-40)
    \Line[width=0.25](0,0)(20,-10)%

    \Line[width=0.25](0,-40)(20,-30)%

    \Line[width=0.25](20,-10)(20,-30)%
    \DoubleLine(0,-40)(0,0){3}%
    \Vertex(0,0){3}
    \Vertex(0,-40){3}
    \Vertex(20,-10){3}
    \Vertex(20,-30){3}
   \end{axopicture}
}

\def\TTsubgraph2{
   \begin{axopicture}(30,30)(-15,-3)
    \SetWidth{1.0}
    \SetColor{black}
	\SetScale{0.2}
	\Arc[width=0.25](0,0)(40,90,270)
	\Vertex(0,-40){10}
	\Vertex(0,40){10}
	\Vertex(40,0){10}
	\Vertex(-40,0){10}
	\Line[width=0.25](0,-40)(0,40)
	\Line[width=0.25](40,0)(0,40)
	\Line[width=0.25](40,0)(0,-40)
\end{axopicture}
}

\def\TAsubgraph3{
   \begin{axopicture}(30,30)(-15,-3)
    \SetWidth{1.0}
    \SetColor{black}
	\SetScale{0.2}
	\ECirc(0,0){40}
	\Vertex(40,0){10}
	\Vertex(-40,0){10}
\end{axopicture}
}

 \def\Tsubgraph1{
   \begin{axopicture}(30,30)(-15,-3)
     \SetWidth{1.0}
     \SetColor{black}
 	\SetScale{0.2}
 	\ECirc(0,0){40}
 	\Vertex(0,-40){10}
 \end{axopicture}
 }

\def\TSsubgraph4{
   \begin{axopicture}(30,30)(-15,-3)
    \SetWidth{1.0}
    \SetColor{black}
	\SetScale{0.2}
	\Arc[width=0.25](0,0)(40,90,270)
	\Vertex(0,-40){10}
	\Vertex(0,40){10}
	\Vertex(40,0){10}
	\Vertex(-40,0){10}
	\Line[width=0.25](40,0)(0,40)
	\Line[width=0.25](40,0)(0,-40)
\end{axopicture}
}

\def\TOsubgraph4{
   \begin{axopicture}(30,30)(-15,-3)      
    \SetWidth{1.0}
    \SetColor{black}
	\SetScale{0.2}
	\Arc[width=0.25](0,0)(40,90,270)
	\Vertex(0,-40){10}
	\Vertex(0,40){10}
	\Vertex(-40,0){10}
	\Line[width=0.25](0,40)(0,-40)
\end{axopicture}
}

\def\Ttriangle1{
   \begin{axopicture}(30,30)(-6,-18)
	\SetScale{0.8}
    \Line[width=0.25](0,-10)(20,-20)%

    \Line[width=0.25](0,-30)(20,-20)%

    \Line[width=0.25](0,-10)(0,-30)%
    \Vertex(0,-10){2}
    \Vertex(0,-30){2}
    \Vertex(20,-20){2}

	\end{axopicture}
}

\def\Tttriangle2{
   \begin{axopicture}(40,25)(-15,-13)
	\SetScale{0.5}
	\Line[width=0.25](0,0)(20,-10)%
    \Line[width=0.25](20,-10)(40,-20)%

    \Line[width=0.25](0,-40)(20,-30)%
    \Line[width=0.25](20,-30)(40,-20)%

    \Line[width=0.25](-20,10)(0,0)
    \Line[width=0.25](-20,-50)(0,-40)
    \DoubleLine(0,-40)(0,0){3}%
    \Vertex(0,0){3}
    \Vertex(0,-40){3}
    \Vertex(40,-20){3}
    \end{axopicture}
}

\def\michitriangle{
\begin{axopicture}(150,60) (50,0)
    \SetWidth{1.0}
    \SetColor{black}
    \Arc[width=0.25](128,12)(45,0,180)
    \Line[width=0.25](80,12)(176,12)
    \Vertex(83,12){6}
    \Vertex(173,12){6}
    \Vertex(128,57){6}

  \end{axopicture}
}
\def\michitadpole{
\begin{axopicture}(150,60) (50,0)
    \SetWidth{1.0}
    \SetColor{black}
    \Arc[width=0.25](128,12)(45,0,360)
    \Vertex(128,57){6}

  \end{axopicture}
}

\def\diagramtwoloopB{
  	\begin{axopicture}(80,45)(-10,-5)
	\SetWidth{1.0}
    \SetColor{black}    
    \Line[width=0.25](0,0)(20,0)
    \Arc[width=0.25](40,0)(20,0,360)
    \Vertex(20,0){3} %
    \Vertex(60,0){3} %
    \Vertex(40,-20){3} %
    \Bezier(20,0)(32,-5)(32,-5)(40,-19)
    \Line[width=0.25](40,-35)(40,-19)
   \end{axopicture}
}

\def\diagramtwoloopBt{
  	\begin{axopicture}(80,45)(-10,-20)
	\SetWidth{1.0}
    \SetColor{black}    
    \Line[width=0.25](0,0)(20,0)
    \Arc[width=0.25](40,0)(20,0,360)
    \Vertex(20,0){3} %
    \Vertex(60,0){3} %
    \Vertex(40,-20){3} %
    \Bezier(20,0)(32,-5)(32,-5)(40,-19)
    \Line[width=0.25](40,-35)(40,-19)
   \end{axopicture}
}

\def\diagramthreeloopA{
  	\begin{axopicture}(80,40)(0,-5)
	\SetWidth{1.0}
    \SetColor{black}    
    \Line[width=0.25](0,0)(20,0)
    \Arc[width=0.25](40,0)(20,0,360)
    \Line[width=0.25](80,0)(60,0)
    \Vertex(20,0){3} %
    \Vertex(60,0){3} %
    \Vertex(40,20){3} %
    \Bezier(60,0)(45,6)(45,6)(40,20)
    \Vertex(40,-20){3} %
    \Bezier(20,0)(32,-5)(32,-5)(40,-19)
   \end{axopicture}
}
\def\diagramthreeloopB{
  	\begin{axopicture}(80,40)(-5,-10)
	\SetWidth{1.0}
    \SetColor{black}    
    \Line[width=0.25](0,0)(20,0)
    \Arc[width=0.25](40,0)(20,0,360)
    \Vertex(20,0){3} %
    \Vertex(60,0){3} %
    \Vertex(40,20){3} %
    \Bezier(60,0)(45,6)(45,6)(40,20)
    \Vertex(40,-20){3} %
    \Bezier(20,0)(32,-5)(32,-5)(40,-19)
    \Line[width=0.25](40,-35)(40,-19)
   \end{axopicture}
}

\def\twoloopunderA{
	\begin{axopicture}(100,40)(-10,-5)
	\Line[width=0.25](0,0)(20,0)
    \Arc[width=0.25](40,0)(20,0,360)
    \Line[width=0.25](80,0)(60,0)
    \Vertex(20,0){3} %
    \Vertex(60,0){3} %
    \Vertex(40,-20){3} %
    \Bezier(20,0)(32,-5)(32,-5)(40,-19)
   \end{axopicture}
}    

\def\vacuumoneloop{
	\begin{axopicture}(60,40)(-10,-6)
    \Arc[width=0.25](20,0)(20,0,360)
    \Vertex(0,0){3} %
    \Vertex(40,0){3} %
    \end{axopicture}
}

\def\bubble{
	\begin{axopicture}(90,40)(-5,-5)
	\Line[width=0.25](0,0)(20,0)
    \Arc[width=0.25](40,0)(20,0,360)
    \Line[width=0.25](80,0)(60,0)
    \Vertex(20,0){3} %
    \Vertex(60,0){3} %
    \end{axopicture}
}

\def\bubbleOneTop{
	\begin{axopicture}(50,40)(-5,-5)
    \Arc[width=0.25](20,0)(20,0,360)
    \Vertex(20,20){3} %
    \end{axopicture}
}

\def\bubbleHalfTop{
	\begin{axopicture}(90,40)(-5,-5)
	\Line[width=0.25](0,0)(20,0)
    \Arc[width=0.25](40,0)(20,0,180)
    \Line[width=0.25](80,0)(60,0)
    \Vertex(20,0){3} %
    \Vertex(60,0){3} %
    \end{axopicture}
}

\def\bubbleOneBottom{
	\begin{axopicture}(50,40)(-5,-5)
    \Arc[width=0.25](20,0)(20,0,360)
    \Vertex(20,-20){3} %
    \end{axopicture}
}

\def\bubbleHalfBottom{
	\begin{axopicture}(90,40)(-5,-5)
	\Line[width=0.25](0,0)(20,0)
    \Arc[width=0.25](40,0)(20,180,360)
    \Line[width=0.25](80,0)(60,0)
    \Vertex(20,0){3} %
    \Vertex(60,0){3} %
    \end{axopicture}
}

\def\franzone{
	\begin{axopicture}(90,40)(-5,-5)
	\Line[width=0.25](0,0)(20,0)
    \Arc[width=0.25](40,0)(20,180,360)
    \Line[width=0.25](80,0)(60,0)
    \Vertex(20,0){3} %
    \Vertex(60,0){3} %

    \Line[width=0.25](20,0)(60,0)

	\Line[width=0.25](20,0)(20,8)
	\Line[width=0.25](60,0)(60,8)

    \fontsize{6}{100}\selectfont
    \Text(37,3)[lb]{2}
    \Text(37,-28)[lb]{3}
    \end{axopicture}
}

\def\franzonevertex{
	\begin{axopicture}(80,40)(-5,-5)

    \Arc[width=0.25](40,0)(20,0,180)

    \Vertex(40,20){3}

    \Arc[width=0.25](20,-3)(3,0,360)
    \Arc[width=0.25](60,-3)(3,0,360)
    \fontsize{6}{100}\selectfont
    \Text(16,12)[lb]{1}
    \Text(57,12)[lb]{4}
    \end{axopicture}
}

\def\diagramtwovacuum{
  	\begin{axopicture}(62,40)(7,-5)
	\SetWidth{1.0}
    \SetColor{black}    
    \Arc[width=0.25](40,0)(20,0,360)
    \Vertex(20,0){3} %
    \Vertex(60,0){3} %
    \Vertex(40,20){3} %
    \Bezier(60,0)(45,6)(45,6)(40,20)
   \end{axopicture}
}

\def\examplelegs23{
\begin{axopicture}(200,60)(31,0)
    \SetWidth{1.0}
    \SetColor{black}
    \Arc[width=0.25](128,12)(45,135,495)
    \Line(32,12)(80,12)
    \Line(224,12)(176,12)
    \Vertex(83,12){6}
    \Vertex(173,12){6}    

     \fontsize{6}{100}\selectfont
    \Text(123,-57)[lb]{2}
    \Text(123,65)[lb]{3}
\end{axopicture}
}

\def\moticexample{
\begin{axopicture}(80,20)(0,-5)
    \SetWidth{1.0}
    \SetColor{black}    
    \Line(0,0)(20,0)
    \Arc[double](40,0)(20,0,180)
    \Arc(40,0)(20,180,360)
    \Line(80,0)(60,0)
    \Vertex(20,0){3} %
    \Vertex(60,0){3} %
    \Line(20,0)(60,0)	
   \Text(38,-18)[lb]{\tiny 3\tiny }
   \Text(38,2)[lb]{\tiny  2\tiny }
   \Text(38,23)[lb]{\tiny 1\tiny }
\end{axopicture}
}

\def\moticexamplea{
\begin{axopicture}(80,30)(0,-5)
    \SetWidth{1.0}
    \SetColor{black}    
    \Line(0,0)(20,0)
    \Arc[double](40,0)(20,0,180)
    \Line(80,0)(60,0)
    \Vertex(20,0){3} %
    \Vertex(60,0){3} %
   \Text(38,23)[lb]{\tiny 1 \tiny }
\end{axopicture}
}

\def\moticexampleab{
\begin{axopicture}(80,30)(0,-5)
    \SetWidth{1.0}
    \SetColor{black}    
    \Line(0,0)(20,0)
    \Arc[double](40,0)(20,0,180)
    \Line(80,0)(60,0)
    \Vertex(20,0){3} %
    \Vertex(60,0){3} %
    \Line(20,0)(60,0)	
   \Text(38,23)[lb]{\tiny 1 \tiny }
   \Text(38,2)[lb]{\tiny 2 \tiny }
\end{axopicture}
}

\def\moticexampleac{
\begin{axopicture}(80,30)(0,-5)
    \SetWidth{1.0}
    \SetColor{black}    
    \Line(0,0)(20,0)
    \Arc[double](40,0)(20,0,180)
    \Arc(40,0)(20,180,360)
    \Line(80,0)(60,0)
    \Vertex(20,0){3} %
    \Vertex(60,0){3} %
   \Text(38,23)[lb]{\tiny 1 \tiny }
   \Text(38,-18)[lb]{\tiny 3 \tiny }
\end{axopicture}
}

\def\moticexamplebc{
\begin{axopicture}(80,30)(0,-5)
    \SetWidth{1.0}
    \SetColor{black}    
    \Arc(40,0)(20,180,360)
    \Vertex(20,0){3} %
    \Vertex(60,0){3} %
    \Line(20,0)(60,0)	
   \Text(38,-18)[lb]{\tiny 3\tiny }
   \Text(38,2)[lb]{\tiny 2 \tiny }
\end{axopicture}
}

\title{The Hopf algebra structure of the $R^*$-operation}

\author{Robert Beekveldt$^{1}$, Michael Borinsky$^{1}$, Franz Herzog$^{1,2}$ \\ 
}
\date{%
    $^1$\small{\textsc{Nikhef Theory Group, Science Park 105, 1098 XG Amsterdam, The Netherlands}}\\%
    $^2$\small{\textsc{Higgs Centre for Theoretical Physics, School of Physics and Astronomy ---\\ --- The University of Edinburgh, Edinburgh EH9 3FD, Scotland, UK}}%
}

\begin{document}

\vspace*{-2\baselineskip}%
\hspace*{\fill} \mbox{\footnotesize{\textsc{Nikhef 2020-005}}}

{\let\newpage\relax\maketitle}

\begin{abstract}
We give a Hopf-algebraic formulation of the $R^*$-operation, which is a canonical way to render UV and IR divergent Euclidean Feynman diagrams finite. Our analysis uncovers a close connection to Brown's Hopf algebra of motic graphs. Using this connection we are able to provide a verbose proof of the long observed `commutativity' of UV and IR subtractions. We also give a new duality between UV  and IR counterterms, which, entirely algebraic in nature, is formulated as an inverse relation on the group of characters of the Hopf algebra of log-divergent scaleless Feynman graphs. Many explicit examples of calculations with applications to infrared rearrangement are given.
\end{abstract}

\section{Introduction}

Feynman integrals are the fundamental building blocks of quantum field theory multi-loop scattering amplitudes which describe the interactions of fundamental particles measured for instance at particle colliders such as the Large Hadron Collider. As such their mathematical properties are directly connected to the physical properties of the scattering amplitudes and are important to understand both from a conceptual as well as from a practical point of view. By better understanding these properties one can for instance improve the methodology for performing much-needed multi-loop precision calculations. In recent decades it has been discovered that Hopf algebras, which operate on the vector space of the Feynman graphs, play a central role in this connection.

The first such example is the Connes-Kreimer Hopf algebra of Feynman diagrams \cite{Kreimer:1997dp,Connes:1999yr}. It relates physically to the \textit{ultraviolet} (UV) divergences and their removal via the procedure of renormalisation. This Hopf algebra governs in particular the combinatorics of the forest formula, a rigorous, but also practical, formulation of renormalisation developed by Bogoliubov, Parasiuk, Hepp and Zimmermann (BPHZ) already more than 50 years ago \cite{Bogoliubov:1957gp,Hepp:1966eg,Zimmermann:1969jj}. The Hopf-algebraic formulation has allowed to organise the algebraically somewhat ad-hoc action of the forest formula through the natural Hopf-algebraic operations: the coproduct and the antipode. Thereby many proofs and properties related to renormalisation could be drastically simplified. 

Another example of such a Hopf algebra, due to Goncharov, is related to generalised polylogarithms and multiple zeta values \cite{Goncharov:2001iea,Brown:2011ik,Duhr:2012fh}, which govern the analytic properties and functional relations of a class of special functions to which Feynman integrals often evaluate. 

Hopf algebras also play an important role while analysing the surprisingly rich number theoretical properties of Feynman integrals \cite{Bloch:2005bh,Brown:2009ta} with their underlying \emph{motives}, a generalised notion of cohomology. For instance, these deep insights into the structure of Feynman integrals were used to prove that multiple zeta values and multiple polylogarithms are insufficient to describe Feynman amplitudes in general, even in the massless context as was initially conjectured by Kontsevich \cite{belkale2003,brown2012}. Another fruit of this program was the development and implementation of an extremely general analytic integration algorithm for Feynman integrals \cite{Brown:2008um,Panzer:2014caa}. Most recently, this program culminated in the formulation of the conjecture that Feynman amplitudes fulfill a general \emph{coaction principle} and are heavily constrained by a number theoretical symmetry encoded in the so called \emph{cosmic Galois group} \cite{Brown:2015fyf,Schnetz:2016fhy}. Such a \emph{coaction}, which is a generalised cousin of a \emph{coproduct}, will reduce to a coaction which is related to the Goncharov coproduct when restricted to motivic multiple zeta valued amplitudes \cite{brown2012mixed,Brown:2015fyf}. More recently, it has even been conjectured that such a coaction can be realised combinatorially at the Feynman graph level \cite{Abreu:2017enx,Abreu:2017mtm,Abreu:2019eyg}. 

Here, we will make heavy use of the Hopf algebra of \emph{motic graphs}, which was first formulated by Brown \cite{Brown:2015fyf}, and repurpose it to deal explicitly with the long-distance or \textit{infrared} (IR) divergences  which can occur in Euclidean Feynman diagrams. This Hopf-algebraic construction can be seen as a generalisation of the  Connes-Kreimer Hopf algebra of UV renormalisation.  A generalisation of the BPHZ forest formula to Euclidean IR divergences, known as the $R^*$-operation, has already been given about 35 years ago by Chetyrkin, Tkachov and Smirnov \cite{Chetyrkin:1982nn,Chetyrkin:1984xa}; see also \cite{Caswell:1981ek,Smirnov:1986me,Phi4,Larin:2002sc,Chetyrkin:2017ppe}. This generalisation is not only of conceptual interest, but also relevant for practical calculations.

\paragraph{Why study Euclidean IR divergences?}While IR divergences do not usually occur in physically relevant Euclidean correlators, they can appear after an infrared rearrangement (a rearrangement of the external momenta) of the graph has been performed. The IR rearrangement procedure is useful as it allows one to simplify the calculation of renormalisation counterterms \cite{Vladimirov:1979zm}. In this context the $R^*$-operation comes in handy, as it can be used to subtract the unphysical IR divergences, which were introduced by an IR rearrangement. In fact, the local counterterm of an $L$-loop $n$-point Feynman diagram can, with the aid of the $R^*$-operation, be extracted from a set of $2$-point Feynman diagrams each with at most $(L-1)$ loops. The $R^*$-operation has been used extensively in calculations of anomalous dimensions in scalar $\phi^4$ theory up to six loops \cite{Chetyrkin:1981jq,Kleinert:1991rg,Kompaniets:2016hct,Kompaniets:2017yct,Batkovich:2016jus}, which were only exceeded by a seven loop result obtained with position space techniques \cite{Schnetz:2016fhy}, and in QCD up to five loops \cite{Herzog:2017ohr,Herzog:2018kwj,Herzog:2017dtz}. In particular the latter developments define the current state of the art in QCD and were possible in part through an extension of the local $R^*$-operation to generic Feynman diagrams with tensor structures attached to them \cite{Herzog:2017bjx}. Another so-called global formulation of the $R^*$-operation, which derives IR counterterms at the level of Green's function by making use of the heavy-mass expansion, has been developed and applied in many important multi-loop calculations; see e.g. \cite{Baikov:2012er,Baikov:2014qja,Baikov:2005rw,Baikov:2016tgj,Chetyrkin:2017bjc}. 

Since the $R^*$-operation leads in essence to a generalised forest formula, one can expect that there should exist a Hopf-algebraic formulation just like in the UV case. To provide such a formulation constitutes the main aim of this paper. A big step in this direction has in fact already been taken by Brown \cite{Brown:2015fyf}, who showed that all singularities of non-exceptional Euclidean Feynman graphs (this includes UV and IR divergences in Feynman parametric space as characterised originally by Speer \cite{Speer:1975dc}) can be described by a generalisation of the Connes-Kreimer Hopf algebra. The basic building blocks of this generalised Hopf algebra are the so-called \textit{motic} graphs. While Brown's construction is certainly related to the $R^*$-operation, it differs from it in two essential ways: it is based on the Feynman parametric representation and the divergences are isolated through blow-ups of certain linear subspaces in the projectified Feynman parametric space. Brown's goal with this setup was to define a \textit{motivic Feynman integral}. The $R^*$-operation is instead formulated in momentum space and relies on dimensional regularisation to isolate divergences. A striking difference between the two formulations is that in the $R^*$-operation a clear distinction is made between IR and UV divergences, which are treated with individual counterterm operations. In Brown's formulation these two operations are instead merged into one, which despite of its beautiful simplicity does obscure their physical origin. 

\paragraph{Outline of this paper}
The main goal of this work is to rewrite the $R^*$-operation in terms of the operations of the Hopf algebra of motic subgraphs. We start with a review of the various topics involved: the local $R^*$-operation, motic subgraphs and their associated Hopf algebra in Section~\ref{section:background}. In Section~\ref{sec:R*Hopf}, we begin with a proof for a close relationship between IR-subgraphs, as they were defined in the original works \cite{Chetyrkin:1984xa,Smirnov:1986me}, and a subclass of motic subgraphs. With this identification we are then able to give a complete formulation of the $R^*$-operation in the language of the Hopf-algebra of motic subgraphs. In particular, we introduce two coactions which describe the combinatorics of the UV and IR divergences respectively. We then reformulate the IR and UV counterterm operations using these coactions. Detailed examples of the formalism are given in Section~\ref{sec:IRrearrangement} where the counterterm calculation of a Feynman diagram via IR rearrangement is performed.
In Section~\ref{sec:antipode} we uncover a duality between the UV and IR counterterm maps that becomes apparent in the Hopf algebra framework. This very simple duality can serve as a definition of the IR counterterm operation. We conclude with a discussion of future uses and extensions of the proposed formalism in Section~\ref{sec:conclusion}.

\section{Background}
\label{section:background}

Our discussion will be focused on Euclidean Feynman integrals whose external momenta $p_{i}$ with $i\in\{1,..,n\}$ are \textit{non-exceptional}. That means any non-trivial subset of the vectors $p_i$ sums to a non-null vector:
\beq
\sum_{i\in S} p_i\neq 0, \text{ for any } S\subsetneq \{1,..,n\}.
\eeq
All divergences of such Feynman integrals can be handled by the $R^*$-operation. Before turning to the $R^*$-operation we will first review the two separate operations on which it is based. These are the $R$-operation, which describes the renormalisation of UV divergences at the level of individual Feynman diagrams and its infrared counterpart the $\t{R}$-operation. 

\subsection{\texorpdfstring{The $R$-operation and divergent UV-subgraphs}{The R-operation and divergent UV-subgraphs}}

\label{setup:R}

The classic $R$-operation was originally formulated by Bogoliubov, Parasiuk \cite{bogoliubov1957multiplication} and later corrected by Hepp \cite{hepp1966proof}. A generalised formulation by Zimmermann followed \cite{Zimmermann:1968mu,Zimmermann:1969jj}. We highly recommend Collin's book on renormalisation for an in depth account on the $R$-operation and renormalisation in general \cite{collins1985renormalization}. Here, we shall follow the notation of \cite{Herzog:2017bjx}. In this reference as well as in earlier works \cite{Caswell:1981ek,Phi4,Chetyrkin:1982nn,Chetyrkin:1984xa} the symbol $\Delta$ has been used to denote local counterterm operations, this notation unfortunately clashes with the standard notation used in the Hopf algebra literature, where $\Delta$ is used for the coproduct. To avoid confusion we shall instead use the symbol $Z$ to describe the local counterterm operations and reserve $\Delta$ for the Hopf algebra's coproduct or associated coactions. To physicists $Z$ is indeed familiar since it is commonly used to describe the renormalisation constants appearing in quantum field theory Lagrangians, which are directly related to the local counterterm. 

A second point which deserves clarification concerns the notion of a Feynman graph $\Gamma$ versus its associated Feynman integral, which is obtained after applying the Feynman rules to $\Gamma$. In the physics literature these two notations are often used interchangeably, and it is then left to the context to determine whether $\Gamma$ refers to the graph or its integral. For example,
\begin{align*} \bubbleonenum{}{} = \int \frac{d^Dk}{\pi^{D/2}} \frac{1}{k^2(k+p)^2} \end{align*}
with $p$ the external momentum and $D$ the dimension of spacetime. 

While we shall use this physicist's convention in our review in this section, we will be more strict on this point in the Hopf-algebraic formulation presented in Section~\ref{sec:R*Hopf}. There, we will use the notation $\Gamma$ exclusively for the respective Feynman diagram and the verbose notation $\phi(\Gamma)$ for the associated Feynman integral, where $\phi$ denotes the Feynman rules interpreted as an explicit map from combinatorial diagrams to integrals.

Another important property of a Feynman diagram is its superficial degree of divergence $\omega(\Gamma)$:
\begin{align*} \omega(\Gamma)= D L_\Gamma + \sum_{v \in V_\Gamma} \omega(v)-\sum_{e \in E_\Gamma} \omega(e), \end{align*}
where $V_\Gamma$ is the set of vertices of $\Gamma$, $E_\Gamma$ is its set of edges and $L_\Gamma$ is its number of loops. For vertices and edges, $\omega(v)$ and $\omega(e)$ stand respectively for the mass-dimension of their Feynman rule. For concreteness, we will assume the usual scalar case, $\omega(v) = 0$ and $\omega(e)=2$ with $D=4$ in all given examples. The statements remain true for arbitrary edge and vertex weights. 
As the name suggests the superficial degree of divergence characterises the asymptotic UV, or large loop-momentum, behaviour of the Feynman integrand. The Feynman integral is said to be superficially UV divergent if $\omega(\Gamma)\ge0$, while a scaleless graph is said to be superficially IR-divergent, i.e.\ it diverges in the limit of small loop-momenta, if $\omega(\Gamma)\le0$. Here we refer to a \emph{vacuum graph} as a Feynman graph which has no dependence on external momenta; a \emph{scaleless graph} is a vacuum graph which also has no internal masses.

The $R$-operation acting on a Feynman graph $\Gamma$, is then defined as follows:
\begin{align} \label{eq:defR} R(\Gamma):= \sum_{\gamma} Z(\gamma)*\Gamma / \gamma, \end{align}
where we sum over all \textit{divergent UV-subgraphs} $\gamma$. We will call a, possibly disconnected, subgraph a \textit{UV-subgraph} if each of its connected components is one-particle-irreducible (1PI). The UV-subgraph is a \textit{divergent UV-subgraph} if each of its connected components is UV superficially divergent. The empty subgraph is also considered to be a divergent UV-subgraph.

We can write any subgraph $\gamma$ appearing in the sum as a disjoint union of its connected components: $\gamma= \bigsqcup_i \delta_i$ such that all $\delta_i$ are UV superficially divergent 1PI Feynman graphs, $\omega(\delta_i) \geq 0$. The contracted graph $\Gamma/\gamma$ represents the contraction of each component $\delta_i$ of $\gamma$ to a single vertex in $\Gamma$. The counterterm operation $Z$ for IR finite graphs $\Gamma$ is given by
\begin{align} Z(\Gamma)= -K\left( \sum_{\gamma \neq \Gamma} Z(\gamma)*\Gamma / \gamma \right), \end{align}
where we sum over all \textit{divergent UV-subgraphs} $\gamma$ which are not equal to $\Gamma$ and the operation $K$ isolates the divergent part of the integral by some scheme-dependent prescription. A common such prescription on which most of our discussion is based is minimal subtraction (MS). In the MS scheme $K$ projects out simple and higher poles in the dimensional regulator $\eps=(4-D)/2$, with $D$ the dimension of spacetime.
The $Z$-operation is multiplicative
\begin{align*} Z(\delta_1 \delta_2)=Z(\delta_1)Z(\delta_2), \end{align*}
this property is central for any Hopf-algebraic interpretation of renormalisation.

Another ingredient is the $*$-symbol which is used to indicate insertion of $Z(\delta_i)$ into the vertex, into which $\delta_i$ got contracted, in $\Gamma/\gamma$. In general, $Z(\delta)$ is not only a number, but a homogeneous polynomial in the external momenta of the graph $\delta$ of degree $\omega(\delta)$. In this more general case the $*$-symbol is used to indicate that care must be taken as information of the external momenta of the subgraph $\delta$ must be extracted from the contracted graph $\Gamma/\delta$. In the case $\omega(\delta)=0$ of a logarithmic subdivergence $Z(\delta)$ is just a number and the $*$ can be interpreted as an ordinary multiplication. Finally one defines $Z(\emptyset)=1$ and $\Gamma/\Gamma=1$. 

The $R$-operation then subtracts counterterms associated to divergent UV-subgraphs, which render the UV divergent Feynman graph finite. Let us illustrate this procedure with two simple examples:
\begin{align*}  R \l(\bubbleonenum{1}{2}\r) &= 1 * \bubbleonenum{1}{2} +Z\l( \bubbleonenum{1}{2}\r)*1\,, \\ &\nn\\  R \l(\bubbletwonum{1}{2}{3}{4}\r) &= 1*\bubbletwonum{1}{2}{3}{4}+Z \l(\bubbletwonum{1}{2}{3}{4}\r)*1+Z \l(\bubbleonenum{2}{3}\r) *\bubbleonenum{1}{4}\,,    \end{align*}
where we only have UV finite expressions on the right hand side. Later in the paper, we will also give examples for the explicit evaluation of the counterterms.

\subsection{\texorpdfstring{The $\t{R}$-operation and divergent IR-subgraphs}{The R-tilde-operation and divergent IR-subgraphs}}

\label{setup:Rstar}
While the $R$-operation renders a purely UV divergent Feynman integral finite, its IR counterpart the $\t{R}$-operation renders a purely IR divergent Feynman integral finite. It is defined analogously as 
\begin{align} \label{eq:Rtilde} \t{R}(\Gamma):= \sum_{\t{\gamma}} \t{Z}(\t{\gamma})*\Gamma \backslash \t{\gamma}, \end{align}
where we sum over all \textit{divergent contracted IR-subgraphs} $\t \gamma$ and $\Gamma \backslash \t{\gamma} = \bar{\gamma}$, the \textit{complementary graph}, is the graph $\Gamma$ with all vertices and edges contained in $\t{\gamma}$ deleted. The definition of (contracted) IR-subgraphs will be discussed in detail below.

The IR counterterm $\t{Z}({\Gamma})$ for UV finite graphs, which is non-vanishing only for scaleless graphs $\Gamma$, is given by
\begin{align*} \t{Z}(\Gamma)= -K\Big(\sum_{\t{\gamma}\neq \Gamma} \t{Z}(\t{\gamma})*\Gamma \backslash \t{\gamma}\Big). \end{align*}
Contrary to the UV case, $\t Z(\Gamma)$ is, in general, a homogeneous polynomial in derivative operators $\partial/\partial p_i^\mu(\bullet)|_{p_i=0}$ of degree $-\omega(\Gamma)$ acting on $\Gamma \backslash \t{\gamma}$. This point is discussed in length with many examples in \cite{Herzog:2017bjx}. The $*$ symbol therefore represents the action of this Taylor operator onto the remaining graph. For the special case of vanishing superficial degree of convergence $\omega(\Gamma)=0$, which receives most of our attention, the $*$ symbol however simply reduces to standard multiplication.

\paragraph{Definition of IR-subgraphs}
Roughly following \cite{Chetyrkin:1984xa}, let us consider a subgraph $\gamma'$ which consists of a subset of the internal massless edges of $\Gamma$ and all vertices of $\Gamma$ which are incident to only these edges.  Associated to $\gamma'$ is its \textit{complementary graph} $\bar{\gamma}=\Gamma\setminus\gamma'$ which consists of the original graph $\Gamma$ with all the edges and vertices of the subgraph $\gamma'$ deleted. In turn the \textit{contracted graph} is then defined as $\t{\gamma}=\Gamma/\bar{\gamma}$. 
Note that $\gamma'$ and $\t{\gamma}$ have the same edge sets. These definitions are best illustrated in the following example:
\begin{align} \SetScale{0.3} \Gamma&= \SetScale{0.3}\fullext \quad \supset \quad \gamma'= \SetScale{0.6}\franzonevertex,\nn\\ \bar{\gamma}= \Gamma \setminus \gamma^{'}&\SetScale{0.6}= \franzone, \hspace{0.7cm} \t{\gamma}= \Gamma/\bar{\gamma} = \Gamma/(\Gamma \setminus \gamma^{'} )=\SetScale{0.3}\subgraphR14irdiv . \nn \end{align}
Note that the missing (drawn as hollow) vertices of $\gamma'$ are chosen in exactly the way that ensures that $\bar \gamma$ is a proper Feynman graph and that we omitted the incoming legs in the contracted graph $\t \gamma$ as its legs are all incident to the same vertex which is analytically equivalent to no external legs by momentum conservation.
In the following, we will use the notation $\gamma'$, $\bar{\gamma}$ and $\t{\gamma}$ for the respective associated graphs.
For a subgraph $\gamma'$ to be an \textit{IR-subgraph} the following conditions need to be fulfilled:
\begin{enumerate}
 \item [(i)]$\bar \gamma$ contains all massive edges of $\Gamma$,
 \item [(ii)]$\gamma'$ contains no external vertices of $\Gamma$,
 \item [(iii)]each connected component of $\bar \gamma$ is 1PI after identifying\footnote{Two vertices $v_i$ and $v_j$ are identified if they are replaced by a new vertex $v_*$ such that all the edges incident on $v_i$ and $v_j$ are now incident on the new vertex $v_*$.} 
 its external vertices and contracting all massive edges.  
\end{enumerate}
Similar conditions have been first formulated in \cite{Chetyrkin:1984xa}.

In analogy to the concept of 1PI graphs in the UV case, we define an IR-subgraph $\gamma'$ to be \textit{infrared-irreducible} (IRI) if $\t{\gamma}$ is \textit{biconnected} that means it cannot be disconnected by deleting any one vertex. In the physics literature the term \emph{1-vertex-irreducible} is also occasionally used instead of biconnected which is an established notion in the mathematical literature. Any contracted IR-subgraph $\t{\gamma}$ decomposes into a set of mutually edge-disjoint complementary graphs $\t{\gamma} = \t{\delta_1}\cup\cdots\cup\t{\delta_n}$ where each of the associated IR-subgraphs $\delta_1', \ldots, \delta_n'$ is IRI. We say an IR-subgraph $\gamma'$ is divergent if each of the components of the complementary graph $\t{\delta_1}, \ldots, \t{\delta_n}$ is superficially IR divergent, $\omega(\t\delta_k) \leq 0$.

A simple example for the $\t R$-operation is then given by
\vspace*{-0.4cm}
\begin{align*} \t{R}(\SetScale{0.3}\fullext)=\t{Z}(\SetScale{0.3}\subgraphR14irdiv) * \SetScale{0.6}\franzone + 1*\SetScale{0.3}\fullext\, ,\\\nn \vspace*{-0.1cm} \end{align*}
where now the right hand side is IR finite. Later, we will discuss how the counterterm can be evaluated explicitly.
\subsection{\texorpdfstring{The $R^*$-operation}{The R*-operation}}
Having discussed both IR- and UV-subgraphs we are now in a position to define the $R^*$-operation, which subtracts divergences associated to both of them and is valid for graphs that have both, UV and IR divergences:
\begin{align} \label{eq:Rstar} R^*(\Gamma):= \sum_{\substack{\gamma, \t{\gamma}\\ \gamma \cap \t{\gamma} = \emptyset}} \t{Z}(\t{\gamma})*Z(\gamma)*\Gamma /\gamma \backslash \t{\gamma}, \end{align}
where the sum goes over all non-overlapping pairs of divergent UV- and contracted IR-subgraphs $(\gamma,\t{\gamma})$. The requirement that divergent IR- and UV-subgraphs shall not overlap follows directly from the fact that the loop momentum associated to a given edge may not be both large and small at the same time. In the presence of both UV and IR divergences the UV  and IR counterterm operations are generalised by:
\begin{align} \label{eq:defZRstar} Z(\Gamma):= -K\Big( \sum_{\substack{\gamma \neq \Gamma,\t{\gamma}\\ \gamma \cap \t{\gamma} = \emptyset}} \t{Z}(\t{\gamma})*Z(\gamma)*\Gamma /\gamma\backslash \t{\gamma}\Big)\,,\\ \label{eq:deftZRstar} \t Z(\Gamma):=-K\Big( \sum_{\substack{\gamma,\t{\gamma} \neq \Gamma\\ \gamma \cap \t{\gamma} = \emptyset}} \t{Z}(\t{\gamma})*Z(\gamma)*\Gamma /\gamma\backslash \t{\gamma}\Big)\,, \end{align}
where $\Gamma$ is now allowed to have both IR and UV subdivergences.
As an example of the $R^*$-operation we consider the log-divergent graph
\vspace*{0.1cm}
\begin{align*} \Gamma&= \bubbletwodotnumc{1}{2}{3}{4}, \end{align*}
\vspace*{0.1cm}
this graph contains the UV-subgraphs $\gamma_1,\gamma_2$ and the contracted IR-subgraph $\t{\gamma}$
\vspace*{0.1cm}
\begin{gather*} \gamma_1=\bubbletwodotnumc{1}{2}{3}{4},\hspace{1cm} \gamma_2 =\bubbleonenum{1}{4}, \hspace{1cm}\t{\gamma}=\vaconenum{2}{3}. \end{gather*}
\vspace*{0.1cm}
This leads to the following expression of $R^*(\Gamma)$:\\
\vspace*{-0.4cm}
\begin{align*} R^*\bubbletwodotnumc{1}{2}{3}{4}&=&&1*1*\bubbletwodotnumc{1}{2}{3}{4}+\t Z \left( \vaconenum{2}{3} \right)*1*\bubbleonenum{1}{4}+\t Z \left( \vaconenum{2}{3} \right)* Z\l(\bubbleonenum{1}{4}\r)*1\nn\\ &&&\nn\\ &&& +1*Z\l(\bubbleonenum{1}{4}\r)*\tadpolenumrem{2}{3}+1*Z\l(\bubbletwodotnumc{1}{2}{3}{4}\r)*1 \\ &&&\nn\\ &=&&\bubbletwodotnumc{1}{2}{3}{4}+\t Z \left( \vaconenum{2}{3} \right) \bubbleonenum{1}{4}+\t Z \left( \vaconenum{2}{3} \right) Z\l(\bubbleonenum{1}{4}\r) +Z\l(\bubbletwodotnumc{1}{2}{3}{4}\r)\,,\nn  \end{align*}
where we have carried out all the $*$-operations in the last line, which resulted in trivial multiplications, because all divergences were only of logarithmic degree. Note that the fourth term in the first line vanishes since its contracted graph is scaleless and scaleless graphs are assumed to vanish under the application of the Feynman rules (see discussion in Section~\ref{setup:IRR}).

\subsection{Infrared rearrangement}
\label{setup:IRR}

The primary practical purpose of the $R^*$-operation is to extend the \textit{infrared rearrangement} procedure to infrared divergent graphs. The basic idea of an infrared rearrangement is to simplify the calculation of the UV counterterm in the MS-scheme by first setting all external momenta and internal masses to zero and then reinserting two external momenta in such a way as to maximally simplify the resulting Feynman integral. It was first observed by Vladimirov that this technique can drastically simplify the evaluation of IR finite graphs \cite{Vladimirov:1979zm}. This procedure works, because the UV counterterm of a log-divergent graph in the MS-scheme is independent of the external momenta or internal masses \cite{Collins:1974da,Caswell:1981ek}. For graphs with higher superficial degree of divergence, the MS UV counterterm can still be obtained from log-divergent ones by Taylor expanding the integrand before integration. Since the validity of the infrared rearrangement procedure is an essential ingredient to our considerations, we state it here as an explicit assumption on the given set of Feynman rules, renormalisation and regularisation scheme:

\begin{assumption}
\label{asp:firstIR}
If $\Gamma$ is a logarithmically divergent Feynman diagram and $\Lambda$ carries the same topology and propagator powers as $\Gamma$, but does not have any external legs or internal masses and is therefore scaleless, then $\Gamma$ and $\Lambda$ evaluate to the same UV counterterm:
\beq
Z(\Gamma) = Z(\Lambda)\,.
\eeq
\end{assumption}
It is well-known that this assumption is fulfilled in the minimal subtraction scheme. However, we are not aware of an explicit proof.

We further assume that scaleless Feynman graphs $\Lambda$ vanish under the application of the Feynman rules $\phi(\Lambda) = 0$. In dimensional regularisation it has been shown that this can be seen as a coherent definition of the value of the Feynman integral of a scaleless graph \cite{Leibbrandt:1975dj}. Note that $\phi(\Gamma)=0$ neither implies nor is implied by $Z(\Gamma) =0$ or $\t Z(\Gamma) = 0$. In practice, one is always forced to reintroduce a scale into a graph when computing its UV counterterm $Z(\Lambda)$. A choice which results into a particularly simple Feynman integral is the insertion of external momenta into two vertices which are connected by a single edge. The calculation of an $L$-loop UV counterterm can then be extracted from massless propagator-type integrals of at most $L-1$ loops. This trick is originally due to Chetyrkin and Tkachov \cite{Chetyrkin:1982nn}. Its essence is the following identity on Feynman integrals which only have incoming momenta on two adjacent vertices:
\begin{align*} \begin{tikzpicture}[baseline={([yshift=-.6ex]current bounding box.center)}] \coordinate (v0); \coordinate[right=1 of v0] (v1); \coordinate[above=.5 of v0] (v0u); \coordinate[below=.5 of v0] (v0d); \coordinate[below=.2 of v0] (v0m); \coordinate[above=.5 of v1] (v1u); \coordinate[below=.5 of v1] (v1d); \coordinate[below=.2 of v1] (v1m); \coordinate[below left=.3 of v0] (v0e); \coordinate[below right=.3 of v1] (v1e); \draw (v0) .. controls (v0d) and (v1d) .. (v1); \filldraw[preaction={fill,white},pattern=north east lines] (v0) .. controls (v0u) and (v1u) .. (v1) .. controls (v1m) and (v0m) .. (v0); \filldraw (v0) circle(1pt); \filldraw (v1) circle(1pt); \draw (v0) -- (v0e); \draw (v1) -- (v1e); \node[below left] at (v0e) {$Q$}; \node[below right] at (v1e) {$-Q$}; \end{tikzpicture} &=\int \frac{d^Dk}{\pi^{D/2}} \left( \frac{1}{(k+Q)^2} \begin{tikzpicture}[baseline={([yshift=-.6ex]current bounding box.center)}] \coordinate (v0); \coordinate[right=1 of v0] (v1); \coordinate[above=.5 of v0] (v0u); \coordinate[below=.5 of v0] (v0d); \coordinate[below=.2 of v0] (v0m); \coordinate[above=.5 of v1] (v1u); \coordinate[below=.5 of v1] (v1d); \coordinate[below=.2 of v1] (v1m); \coordinate[below left=.3 of v0] (v0e); \coordinate[below right=.3 of v1] (v1e); \filldraw[preaction={fill,white},pattern=north east lines] (v0) .. controls (v0u) and (v1u) .. (v1) .. controls (v1m) and (v0m) .. (v0); \filldraw (v0) circle(1pt); \filldraw (v1) circle(1pt); \draw (v0) -- (v0e); \draw (v1) -- (v1e); \node[below left] at (v0e) {$k$}; \node[below right] at (v1e) {$-k$}; \end{tikzpicture} \right) \\ &= \left[ \begin{tikzpicture}[baseline={([yshift=-.6ex]current bounding box.center)}] \coordinate (v0); \coordinate[right=1 of v0] (v1); \coordinate[above=.5 of v0] (v0u); \coordinate[below=.5 of v0] (v0d); \coordinate[below=.2 of v0] (v0m); \coordinate[above=.5 of v1] (v1u); \coordinate[below=.5 of v1] (v1d); \coordinate[below=.2 of v1] (v1m); \coordinate[below left=.3 of v0] (v0e); \coordinate[below right=.3 of v1] (v1e); \filldraw[preaction={fill,white},pattern=north east lines] (v0) .. controls (v0u) and (v1u) .. (v1) .. controls (v1m) and (v0m) .. (v0); \filldraw (v0) circle(1pt); \filldraw (v1) circle(1pt); \draw (v0) -- (v0e); \draw (v1) -- (v1e); \node[below left] at (v0e) {$k$}; \node[below right] at (v1e) {$-k$}; \end{tikzpicture} \right]_{k^2 =1 } \int \frac{d^Dk}{\pi^{D/2}} \frac{1}{(k+Q)^2(k^2)^{1+(L-1)\eps}}\nn   \end{align*}
where the dashed region stands for an arbitrary (sub)graph and we have used dimensional analysis to get to the final equality. In the last line we are only left with an $(L-1)$-loop propagator-type Feynman integral and an integral that can be evaluated using $\Gamma$-functions.

\subsection{Motic subgraphs}
\label{setup:Motic}

Brown \cite{Brown:2015fyf} introduced the notion of motic subgraphs which correspond to singularities of the Feynman integrals in \emph{parametric space}.
Conceptually, motic graphs are a natural generalisation of 1PI graphs \cite[Thm 3.6]{Brown:2015fyf}.

Recall that every subset of edges $E_\gamma \subset E_\Gamma$ of a graph $\Gamma$ gives rise to a distinguished subgraph $\gamma$. The set of subgraphs therefore is in bijection with the set of subsets of edges. Brown declares the subgraph $\gamma \subset \Gamma$ to inherit the external momenta and masses of the parent graph in a non-trivial way. To do so he first distinguishes between different types of subgraphs. A subgraph $\gamma$ is
\begin{itemize}
\item \textit{momentum-spanning} in $\Gamma$ if it has a connected component which contains all external vertices of the parent graph.
\item \textit{mass-spanning} in $\Gamma$ if it contains all massive propagators of the parent graph.
\item \textit{mass-momentum-spanning} in $\Gamma$ if it is both mass- and momentum-spanning.
\end{itemize}
Because of the general recursive nature of renormalisation problems, we also need to consider \textit{subgraphs of subgraphs} and extend the above categorisation to those. This is not trivial because we need to clarify how external vertices and masses are inherited from parent graph to subgraph. 

Brown defines, that a subgraph inherits all masses and external momenta of the parent graph if and only if it is mass-momentum-spanning. If the subgraph is not mass-momentum-spanning, it inherits no external momenta and has zero mass at each edge. Note that a slightly counter-intuitive situation arises from this: if $\gamma\subset \Gamma$ is a subgraph which is not mass-momentum-spanning in $\Gamma$, then all subgraphs $\mu \subset \gamma$ are mass-momentum-spanning in $\gamma$. Recall that a scaleless graph does not depend on any external momenta or masses. Therefore, every subgraph of a scaleless graph is mass-momentum-spanning. 

A subgraph $\gamma \subset \Gamma$ is then defined to be \textit{motic} if, for every proper subgraph of the subgraph $\mu \subsetneq \gamma$, which is mass-momentum-spanning in $\gamma$, $\mu$ has strictly less loops than $\gamma$.

Note that every 1PI subgraph is also motic, because it only has proper subgraphs which have strictly less loops.

Consider for instance the graph
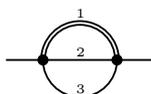
\begin{align*} \SetScale{0.7} \moticexample \end{align*}
where we have labeled the internal edges with numbers and massive propagators are depicted as double edges.
It has the following set of six motic subgraphs:
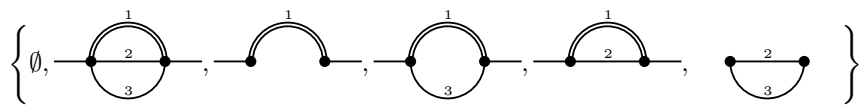
\begin{gather*} \SetScale{0.7} \Bigg\{ \emptyset, \SetScale{0.7}\moticexample, \moticexamplea, \moticexampleac, \SetScale{0.7}\moticexampleab, \moticexamplebc\Bigg\} \,, \end{gather*}
where the first subgraph is the empty graph which is not mass-momentum-spanning. 
The following four subgraphs are mass-momentum-spanning since they contain all massive propagators and external momenta. 
The last subgraph is not mass-momentum spanning, because it does not contain all masses and external momenta of the parent graph. Explicitly, it does not contain the massive edge $1$.

\subsection{Hopf algebra of motic subgraphs}

Brown \cite{Brown:2015fyf} used his definition of motic graphs to introduce a Hopf algebra, which generalises the Connes-Kreimer Hopf algebra of renormalisation \cite{Connes:1999yr}. In Section~\ref{sec:R*Hopf}, we will show that this Hopf algebra can be used to describe the combinatorics of the $R^*$-operation. In this setting the $R$-operations and their variations then act as \emph{linear operators} motivating the terminology $R$-operator, which we shall use hereafter.

For a detailed review on renormalisation Hopf algebras we refer the reader to \cite{Manchon:2001bf} for a more algebraic focus and to \cite[Chap.\ 5]{borinsky2018graphs} with a focus on the combinatorial aspects. An introduction to general Hopf-algebraic structures is given in Sweedler's book \cite{sweedler1969hopf}. %

Readers, who are unfamiliar with Hopf algebras but familiar with the $R$-operator, are advised to already look at the examples in Section~\ref{sec:examples_hopf_basic} while reading the following sections and observe the similarity of the Hopf algebraic and the traditional $R$-operators.

Following Brown, we define $\M$ to be a vector space over $\Q$ which is spanned by all motic graphs. An element of $\M$ is therefore always a formal sum over graphs with rational coefficients. Graphs are considered identical if they are isomorphic to each other. As in the Connes-Kreimer Hopf algebra, the multiplication of two elements of $\M$ is defined using the disjoint union of graphs. The multiplication is the unique linear map, \begin{align*} m_\M : \M \otimes \M \to \M, \end{align*}
such that $m_\M( \Gamma_1 \otimes \Gamma_2 ) = \Gamma_1 \sqcup \Gamma_2$ for two graphs $\Gamma_1$ and $\Gamma_2$. We can immediately identify the empty graph as the neutral element under multiplication. It will be denoted as $\one := \emptyset$. With this multiplication the vector space $\M$ becomes an \textit{algebra}.

A coproduct on $\M$ is also defined on graphs,
\begin{align} \label{coproduct1} \Delta_\M (\Gamma) := \sum_{\gamma \subset \Gamma} \gamma \otimes \Gamma/\gamma, \end{align}
where we sum over all motic subgraphs $\gamma$ of $\Gamma$.
The complete contraction $\Gamma/\Gamma$ appearing in the sum is identified with the empty graph or equivalently with the element $\one$.  This definition can again be extended linearly to a map on the full vector space $\Delta_\M:\; \M \to \M \otimes \M$. With the comultiplication $\M$ is now also a \textit{coalgebra}.

A unit on the algebra is defined as the linear map, $u_\M: \Q \to \M$, $q \mapsto q \one$. The counit $\counit_\M: \M \to \Q$ projects onto the element $\one$ of the algebra. It is defined on graphs as
\begin{align*} \counit_\M(\Gamma) &= \begin{cases} 1 \quad \textrm{ if } \Gamma = \emptyset = \one \\ 0 \quad \textrm{ else }\\ \end{cases} \end{align*}
and extended linearly for all elements in $\M$.

An important property of the coproduct is that it is compatible with the multiplication and the counit,
\begin{align*} \Delta_\M \circ m_\M &= m_{\M \otimes \M} \circ ( \Delta_\M \otimes \Delta_\M ) \\ (\id \otimes \counit_\M) \circ \Delta_\M &= (\counit_\M \otimes \id) \circ \Delta_\M = \id, \end{align*}
where $m_{\M \otimes \M}$ is the natural multiplication on the tensor product $\M \otimes \M$ defined on pairs of graphs by $m_{\M \otimes \M}( (\gamma_1 \otimes \mu_1) \otimes (\gamma_2 \otimes \mu_2) ) = m_\M(\gamma_1 \otimes \gamma_2) \otimes m_\M( \mu_1 \otimes \mu_2)$ and  $\id$ is the identity map on $\M$.

Another important property is the coproduct's \textit{coassociativity},
\begin{align*} ( \id \otimes \Delta_\M ) \circ \Delta_\M = ( \Delta_\M \otimes \id ) \circ \Delta_\M, \end{align*}

Moreover, $\M$ is \textit{connected} and therefore there exists an \textit{antipode}, a linear map $S : \M \to \M$ which fulfills,
\begin{align*} \label{eq:Sdef} m_\M \circ ( S \otimes \id ) \circ \Delta_\M = m_\M \circ ( \id \otimes S ) \circ \Delta_\M = \unit_\M \circ \counit_\M, \end{align*}
such that $\M$ becomes a \textit{Hopf algebra}.
Brown proved these algebraic statements for $\M$ in \cite[Thm.\ 4.2]{Brown:2015fyf}. 

The definition of $S$, eq.\ \eqref{eq:Sdef}, gives rise to a recursive formula for the action of $S$ on graphs $\Gamma \neq \one$
\begin{align} S(\Gamma)= - \Gamma - \sum_{\emptyset \neq \gamma \subsetneq\Gamma}S(\gamma) \Gamma/\gamma= - \Gamma - \sum_{\emptyset \neq \gamma \subsetneq\Gamma}\gamma S(\Gamma/\gamma), \end{align}
which terminates with $S(\one) = \one$.
The convergence of this recursive formula is ensured because $\M$ is \textit{graded} by both the loop-number and the number of edges of the graphs.
We give an example of the antipode of a Feynman graph with one massive propagator:
\vspace{-0.2cm}
\begin{gather*} S\Big(\trianglefull\Big)= -\trianglefull - S\Big(\1subgraph1\Big)\Tsubgraph1 -S\Big(\2subgraph2\Big)\TTsubgraph2 \nn\\ - S\Big(\3subgraph3\Big)\TAsubgraph3 -S\Big(\Ttriangle1\Big)\Tttriangle2 -S\Big(\2subgraph2 \Ttriangle1\Big) \TAsubgraph3\\ =-\trianglefull - \Big(-\1subgraph1-\2subgraph2\TSsubgraph4\Big)\Tsubgraph1 +\2subgraph2\TTsubgraph2\nn\\ - \Big(-\3subgraph3-\2subgraph2\TOsubgraph4\Big)\TAsubgraph3 + \Ttriangle1\Tttriangle2 -\2subgraph2 \Ttriangle1 \TAsubgraph3.\nn \end{gather*}
Note that due to the motic property, there is always one unique subgraph that carries the external momentum information of the original graph in the disjoint union of graphs as indicated by the legs connected to the graph. This component is also the one that has to contain the massive edges.

\section{\texorpdfstring{Hopf-algebraic formulation of $R^*$}{Hopf-algebraic formulation of R*}}
\label{sec:R*Hopf}

In this section we reformulate the $R^*$-operation in the language of Brown's Hopf algebra of motic graphs. As reviewed in section \ref{setup:Rstar}, UV divergences are associated with UV-subgraphs which are essentially disjoint unions of 1PI graphs with a superficial divergence criterion, while IR divergences, within the context of $R^*$, are associated with IR-subgraphs, which are more complicated. An alternative arguably simpler definition then the one in Section~\ref{setup:Rstar} of IR-subgraphs can be given in the language of motic subgraphs which we reviewed in the previous section.

\subsection{IR- and motic mass-momentum-spanning subgraphs}

A key observation which we will prove in the following section, is that the set of motic and mass-momentum-spanning subgraphs of a given graph is in one-to-one correspondence with the set of IR-subgraphs. A motic and mass-momentum-spanning subgraph does not directly give rise to an IR-subgraph, but its complementary graph does.

\begin{theorem}
    \label{thm:IRandMotic}
A subgraph $\gamma \subset \Gamma$ is motic and mass-momentum-spanning in $\Gamma$, if and only if its complementary graph is an IR-subgraph of $\Gamma$.
\end{theorem}

\begin{proof}
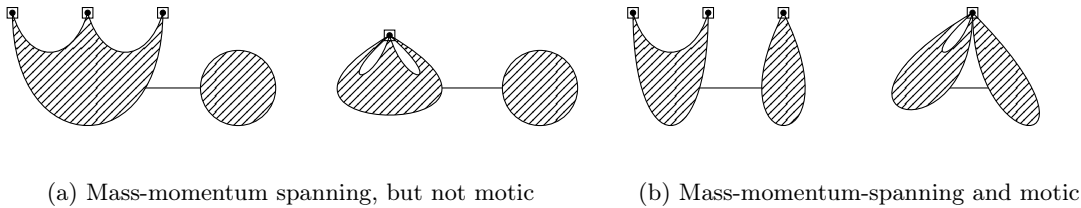
\begin{figure}%
\begin{subfigure}[b]{0.5\textwidth}%
\centering
\begin{tikzpicture} \coordinate (va); \coordinate[right=1 of va] (vb); \coordinate[right=1 of vb] (vc); \coordinate[right=.2 of va] (ca); \coordinate[left=.2 of vb] (cb1); \coordinate[right=.2 of vb] (cb2); \coordinate[left=.2 of vc] (cc); \coordinate[below=.7 of ca] (cca); \coordinate[below=.7 of cb1] (ccb1); \coordinate[below=.7 of cb2] (ccb2); \coordinate[below=.7 of cc] (ccc); \coordinate[below=2 of va] (ccca); \coordinate[below=2 of vc] (cccc); \coordinate[below=1 of vb] (da); \coordinate[right=2 of da] (db); \draw[color=white] (da) circle(1); \draw[color=white] (db) circle(1); \draw (da) -- (db); \draw[preaction={fill,white},pattern=north east lines] (db) circle(.5); \begin{pgfinterruptboundingbox} \filldraw[preaction={fill,white},pattern=north east lines] (va) .. controls (cca) and (ccb1) .. (vb) .. controls (ccb2) and (ccc) .. (vc) .. controls (cccc) and (ccca) .. (va); \end{pgfinterruptboundingbox} \newcommand\Square[1]{+(-#1,-#1) rectangle +(#1,#1)} \draw (va) \Square{2pt}; \draw (vb) \Square{2pt}; \draw (vc) \Square{2pt}; \filldraw (va) circle(1pt); \filldraw (vb) circle(1pt); \filldraw (vc) circle(1pt); \end{tikzpicture}%
\begin{tikzpicture} \coordinate (va); \coordinate[right=1 of va] (vb); \coordinate[right=1 of vb] (vc); \coordinate[right=.2 of va] (ca); \coordinate[left=.2 of vb] (cb1); \coordinate[right=.2 of vb] (cb2); \coordinate[left=.2 of vc] (cc); \coordinate[below=.7 of ca] (cca); \coordinate[below=.7 of cb1] (ccb1); \coordinate[below=.7 of cb2] (ccb2); \coordinate[below=.7 of cc] (ccc); \coordinate[below left=2 of va] (ccca); \coordinate[below right=2 of vc] (cccc); \coordinate[below=.7 of vb] (da); \coordinate[right=2 of da] (db); \draw[color=white] (da) circle(1); \draw[color=white] (db) circle(1); \draw (da) -- (db); \draw[preaction={fill,white},pattern=north east lines] (db) circle(.5); \begin{pgfinterruptboundingbox} \filldraw[preaction={fill,white},pattern=north east lines] (vb) .. controls (cca) and (ccb1) .. (vb) .. controls (ccb2) and (ccc) .. (vb) .. controls (cccc) and (ccca) .. (vb); \end{pgfinterruptboundingbox} \newcommand\Square[1]{+(-#1,-#1) rectangle +(#1,#1)} \draw (vb) \Square{2pt}; \filldraw (vb) circle(1pt); \end{tikzpicture}%
\subcaption{Mass-momentum spanning, but not motic}%
\label{fig:non_motic}
\end{subfigure}%
\begin{subfigure}[b]{0.5\textwidth}%
\centering
\begin{tikzpicture} \coordinate (va); \coordinate[right=.5 of va] (vm); \coordinate[right=1 of va] (vb); \coordinate[right=.2 of va] (ca); \coordinate[left=.2 of vb] (cb1); \coordinate[right=1 of vb] (vc); \coordinate[below=.7 of ca] (cca); \coordinate[below=.7 of cb1] (ccb1); \coordinate[below=2 of va] (ccca); \coordinate[below=2 of vb] (cccb); \coordinate[below=2 of vc] (ccc); \coordinate[left=1 of ccc] (ccc1); \coordinate[right=1 of ccc] (ccc2); \coordinate[below=1 of vm] (da); \coordinate[below=1 of vc] (db); \draw[color=white] (da) circle(1); \draw[color=white] (db) circle(1); \draw (da) -- (db); \begin{pgfinterruptboundingbox} \filldraw[preaction={fill,white},pattern=north east lines] (va) .. controls (cca) and (ccb1) .. (vb) .. controls (cccb) and (ccca) .. (va); \filldraw[preaction={fill,white},pattern=north east lines] (vc) .. controls (ccc1) and (ccc2) .. (vc); \end{pgfinterruptboundingbox} \newcommand\Square[1]{+(-#1,-#1) rectangle +(#1,#1)} \draw (va) \Square{2pt}; \draw (vb) \Square{2pt}; \draw (vc) \Square{2pt}; \filldraw (va) circle(1pt); \filldraw (vb) circle(1pt); \filldraw (vc) circle(1pt); \end{tikzpicture}%
\begin{tikzpicture} \coordinate (va); \coordinate[right=.5 of va] (vm); \coordinate[right=1 of va] (vb); \coordinate[right=.2 of va] (ca); \coordinate[left=.2 of vb] (cb1); \coordinate[right=1 of vb] (vc); \coordinate[below=.7 of ca] (cca); \coordinate[below=.7 of cb1] (ccb1); \coordinate[below left=2 of va] (ccca); \coordinate[below=2 of vb] (cccb); \coordinate[below=2 of vc] (ccc); \coordinate[left=1 of ccc] (ccc1); \coordinate[right=1 of ccc] (ccc2); \coordinate[below=1 of vm] (da); \coordinate[right=.1 of da] (das); \coordinate[below=1 of vc] (db); \coordinate[left=.5 of db] (dbs); \draw[color=white] (da) circle(1); \draw[color=white] (db) circle(1); \draw (das) -- (dbs); \begin{pgfinterruptboundingbox} \filldraw[preaction={fill,white},pattern=north east lines] (vb) .. controls (cca) and (ccb1) .. (vb) .. controls (cccb) and (ccca) .. (vb); \filldraw[preaction={fill,white},pattern=north east lines] (vb) .. controls (ccc1) and (ccc2) .. (vb); \end{pgfinterruptboundingbox} \newcommand\Square[1]{+(-#1,-#1) rectangle +(#1,#1)} \draw (vb) \Square{2pt}; \filldraw (vb) circle(1pt); \end{tikzpicture}
\subcaption{Mass-momentum-spanning and motic}%
\label{fig:motic}
\end{subfigure}
\caption{Illustrations of the relationship between motic mass-momentum-spanning and complementary IR-subgraphs. External vertices are indicated with little squares. The dashed areas stand for arbitrary connected subgraphs.}
\label{fig:motic_vs_non_motic}
\end{figure}
    We are going to start by proving that the complementary graph $\gamma'$ to a motic and mass-momentum-spanning subgraph $\bar{\gamma} \subset \Gamma$ is an IR-divergent subgraph. 

    As $\bar{\gamma} \subset \Gamma$ is mass-momentum-spanning condition (i) and (ii) of the definition of IR-subgraphs in Section~\ref{setup:Rstar} are immediately fulfilled. 

We proceed to prove that a motic and mass-momentum-spanning subgraph fulfills condition (iii). That means it is 1PI after we identify the external vertices.
The relevant situations are illustrated in Figure~\ref{fig:motic_vs_non_motic}. Fig.~\ref{fig:non_motic} depicts a mass-momentum-spanning (sub)graph which is not motic. It is not 1PI and also does not become 1PI after we identify the external vertices. On the other hand, Fig.~\ref{fig:motic} shows a mass-momentum-spanning motic (sub)graph, which is not 1PI, but becomes 1PI after the identification of the external vertices. 

As $\bar{\gamma}$ is mass-momentum-spanning it has one connected component which contains all external vertices. If we identify all external vertices in $\bar{\gamma}$ and $\Gamma$, i.e.\ we pull them together to a single vertex, the number of connected components stays the same. Also, $\bar{\gamma}$ is still a mass-momentum-spanning subgraph of the new graph. By the definition of a motic subgraph, any removal of an edge must either decrease the loop number or make the graph non-mass-momentum-spanning. As all external vertices are identified, mass-momentum-spanning-ness cannot be destroyed anymore. Therefore every removal of an edge must decrease the number of loops of the subgraph. This is condition is equivalent to componentwise 1PIness. 

It remains to be proven that a subgraph $\bar{\gamma} \subset \Gamma$ is motic if the associated subgraph with all external vertices identified is componentwise 1PI. Suppose the opposite and consider a subgraph which is not motic and becomes componentwise 1PI after we identify the external vertices. As the graph is not motic it has an edge whose removal neither destroys mass-momentum-spanning-ness nor reduces the number of loops, see Figure~\ref{fig:non_motic}. If we identify the external vertices of the graph, the removal of this edge still does not decrease the loop number: it could only do so by splitting a loop that was created by the identification of the external vertices, but this would imply breaking mass-momentum-spanning-ness in the original graph. Therefore there is an edge in the graph with identified external vertices whose removal does not decrease the number of loops. This means the identified subgraph cannot be componentwise 1PI, which is a contradiction.
\end{proof}

\subsection{The UV and IR coactions}
Having made the identification of complementary IR-subgraphs with motic mass-momentum-span\-ning subgraphs, we are now in a position to reformulate the graph-combinatorial part of the $R^*$-operator in terms of two new coactions $\Delta^{\UV}_{\M}$ and $\Delta^{\IR}_{\M}$, which are build from the motic coproduct but are restricted to divergent UV- and IR-subgraphs respectively.

The definition of the motic UV-coaction on a Feynman graph $\Gamma$ is,
\begin{align} \label{coproduct3} \Delta_\mathcal{M}^{\text{UV}} (\Gamma) = \sum_{\gamma_{\text{UV}} \subseteq \Gamma} \gamma_{\text{UV}} \otimes \Gamma/\gamma_{\text{UV}} \hspace{8mm} \text{with} \quad \gamma_{\text{UV}}= \bigsqcup_i \delta_i,\quad \omega(\delta_i)\geq0 \end{align}
where we sum over all motic subgraphs, $\gamma_{UV}$, which are also divergent UV-subgraphs. 
This coproduct is nothing, but the Connes-Kreimer coproduct extended in its domain to all motic graphs.

The respective motic IR-coaction on a Feynman graph $\Gamma$ is,
\begin{align} \label{coproduct2} \Delta_\mathcal{M}^{\text{IR}} (\Gamma) = \sum_{\gamma_{\text{mm}} \subseteq \Gamma} \gamma_{\text{mm}} \otimes \Gamma/\gamma_{\text{mm}} \hspace{8mm} \text{with}\quad \Gamma/\gamma_{\text{mm}}= \bigsqcup_i \delta_i,\quad \omega(\delta_i)\leq0 \,, \end{align}
where we now sum over all mass-momentum-spanning (mm) motic subgraphs $\gamma_{\text{mm}}$ which are complementary to a divergent IR-subgraph. Recall that this means that the contracted graph $\Gamma/\gamma_{mm}$ is a scaleless graph which is composed of superficially IR divergent graphs joined in one vertex. Observe that the condition for $\Gamma/\gamma_{\text{mm}}$ to be scaleless is equivalent to the condition of $\gamma_{\text{mm}}$ to be mass-momentum-spanning. Let us remark further that the IR divergences are not `captured' in the left entry, but on the right entry of the tensor product. %

It is now convenient to define two maps, $\rho^{UV}$ and $\rho^{IR}$, which project onto the respective subspaces $\mathcal{H}^{\text{UV}}$ and $\mathcal{H}^{\text{IR}}$,
\begin{align*} &\rho^{UV}: \mathcal{M} \rightarrow \mathcal{H}^{UV}\hspace{2mm}\text{with}\hspace{2mm}\rho^{UV}(\Gamma)=\begin{cases} \Gamma=\bigsqcup_i \delta_i \textrm{ if all } \delta_i \text{ are 1PI and } \omega(\delta_i)\geq 0\\ 0 \textrm{ else}\\ \end{cases}\\ &\rho^{IR}: \mathcal{M} \rightarrow \mathcal{H}^{IR}\hspace{3.8mm} \text{with}\hspace{2mm}\rho^{IR}(\Gamma)=\begin{cases} \Gamma=\bigsqcup_i \delta_i \textrm{ if } \Gamma \text{ is scaleless and }\omega(\delta_i)\leq 0\\ 0 \textrm{ else}.\\ \end{cases} \end{align*}
The respective images, the two subspaces $\H^{\text{UV}}$ and $\H^{\text{IR}}$, are generated by all divergent UV-graphs and all graphs that can appear as contracted IR-subgraphs, i.e., all scaleless and componentwise superficially IR divergent graphs. Clearly, 
\begin{align*} \Delta^{\text{UV}}_\mathcal{M}&=(\rho^{\text{UV}} \otimes \id)\circ\Delta_\mathcal{M} &: &&&\mathcal{M} \rightarrow \mathcal{H}^{\text{UV}} \otimes \mathcal{M}\\ \Delta^{\text{IR}}_\mathcal{M}&=(\id\otimes \rho^{\text{IR}})\circ\Delta_\mathcal{M} &: && &\mathcal{M}\rightarrow \mathcal{M} \otimes \mathcal{H}^{\text{IR}} \end{align*}
In fact, the pair $(\M, \Delta_\mathcal{M}^{\text{UV}})$ is a \emph{left-comodule} over $\mathcal{H}^{\text{UV}}$ and  
$(\M, \Delta_\mathcal{M}^{\text{IR}})$ is a \emph{right-comodule} over $\mathcal{H}^{\text{IR}}$.
We can denote the respective kernels of the maps $\rho^{\text{UV}}$ and $\rho^{\text{IR}}$ as $I_{\text{UV}}:=\ker(\rho^{\text{UV}})$ and $I_{\text{IR}}:=\ker(\rho^{\text{IR}})$. These two subspaces are generated by all graphs that are not divergent UV-graphs or not scaleless componentwise superficially IR divergent graphs respectively. These two subspaces actually form \emph{Hopf ideals} of $\M$. We provide a proof sketch of this fact for readers who are interested in the algebraic details of this construction:
\begin{lemma}

    $I_{\text{UV}}$ and $I_{\text{IR}}$ are Hopf ideals of $\M$.
\end{lemma}
\begin{proof}[Proof sketch]
    Both subspaces are obviously ideals of $\M$ as a disjoint union of some graph with a non-UV divergent graph is still non-UV divergent and for the IR case equivalently. 
    Because $\M$ is a connected Hopf algebra it is sufficient to additionally show that $\Delta_\M I \subset I \otimes \M + \M \otimes I$ for $I$ to be an Hopf ideal.

    With the definition of the motic coproduct in mind, $\Delta_\M \Gamma = \sum_{\gamma \subset \Gamma} \gamma \otimes \Gamma/\gamma$, it is clear that we need to show that 
    the property of `non-UV divergence' or scalefulness needs to be inherited from $\Gamma$ to either $\gamma$ or $\Gamma/\gamma$ for $I_{\text{UV}}$ and $I_{\text{IR}}$ to be Hopf ideals. 

    For the UV case it is easy to verify that a graph $\Gamma$ with a non-1PI component will produce either a $\gamma$ or a $\Gamma/\gamma$ in the coproduct which also has a non-1PI component. The same holds true for a non-superficially divergent 1PI component in $\Gamma$, which must also reappear in either $\gamma$ or $\Gamma/\gamma$.

    In the IR case we need to verify that if $\Gamma$ is scaleful, then either $\gamma$ or $\Gamma/\gamma$ is scaleful. Which is the case by the definition of the motic subgraphs.
\end{proof}
\begin{corollary}
    The quotient Hopf algebras $\M / I_{\text{UV}}$ and $\M / I_{\text{IR}}$ can be identified with $\mathcal{H}^{\text{UV}}$ and $\mathcal{H}^{\text{IR}}$, which are therefore indeed Hopf algebras with the maps, $\Delta_\mathcal{M}^{\text{UV}}$ and $\Delta_\mathcal{M}^{\text{IR}}$ being actual coactions.
\end{corollary}

\subsection{\texorpdfstring{$R$ and $R^*$ in the Hopf-algebraic framework}{R and R* in the Hopf-algebraic framework}}

For the Hopf-algebraic formulation, we are going to strictly differentiate between Feynman graphs and their evaluated value as an integral. 

The \textit{Feynman rules} will be denoted as $\phi: \M \rightarrow \A$, which maps Feynman graphs to an element of some algebra of functions which typically involves a formal parameter $\epsilon$ to incorporate dimensional regularisation. For simplicity, we will mostly assume that $\A$ is equal to the ring of Laurent series in $\epsilon$. This will also be the relevant case in the examples in Section~\ref{sec:IRrearrangement}.

For a gentle introduction into the Hopf-algebraic framework, we will leave the $Z$ and $\t Z$ as originally defined for this section and focus on the Hopf algebra versions of $R$ and $\t R$ for now. In Section~\ref{sec:defIRUVhopf} we will also introduce the Hopf-algebraic versions of $Z$ and $\t Z$.

Especially for readers who are unfamiliar with Hopf-algebraic formulations of renormalisation, we will go through detailed examples of explicit calculations. 

With the $Z$- and $\t Z$-operators defined as before, the $R$- and $\t R$-operator can be written in the Hopf algebra language as,
\begin{align} R(\Gamma)&= m_\A \circ (Z \otimes \phi) \circ \Delta^{\text{UV}}_\mathcal{M} (\Gamma)= \sum_{\gamma_{\text{UV}} \subseteq \Gamma} Z(\gamma_{\text{UV}})\phi(\Gamma/\gamma_{\text{UV}})\,, \\ \widetilde{R}(\Gamma)&= m_\A \circ (\phi \otimes \widetilde{Z}) \circ \Delta^{\text{IR}}_\mathcal{M} (\Gamma)=\sum_{\gamma_{\text{mm}} \subseteq \Gamma}\phi(\gamma_{\text{mm}})\widetilde{Z}( \Gamma/\gamma_{\text{mm}} ). \end{align}
These formulations are equivalent to the ones given in eqs.\ \eqref{eq:defR} and \eqref{eq:Rtilde}. With the equivalence in the second line being non-trivial due to the change from IR-subgraphs to motic mass-momentum-spanning subgraphs.

Since the $R^*$-operator involves both UV and IR divergences, we have to use a combination of both coactions $\Delta^{\text{UV}}_\mathcal{M}$ and $\Delta^{\text{IR}}_\mathcal{M}$ in our Hopf-algebraic framework. An important conceptual prerequisite for that is the following coassociativity-like identity which the coactions $\Delta^{\text{UV}}_\mathcal{M}$ and $\Delta^{\text{IR}}_\mathcal{M}$ fulfill: 
\begin{theorem}
    \label{thm:UVIRcommute}
    \begin{align*} (\id \otimes \Delta^{IR}_\mathcal{M}) \circ \Delta^{UV}_\mathcal{M}=(\Delta^{UV}_\mathcal{M}\otimes \id) \circ \Delta^{IR}_\mathcal{M} = (\rho^{\text{UV}} \otimes \id \otimes \rho^{\text{IR}}) \circ \Delta_\mathcal{M}^{(2)}, \end{align*}
where $\Delta_\mathcal{M}^{(2)} = (\Delta_\M \otimes \id) \circ \Delta_\M = (\id \otimes \Delta_\M) \circ \Delta_\M$ is the double coproduct.
\end{theorem}
This statement gives a precise formulation of the fact that for log-divergent subgraphs it does not matter if either UV  or IR divergences are subtracted first. This `coassociativity' of $\Delta^{UV}_\mathcal{M}$ and $\Delta^{IR}_\mathcal{M}$ is essential for the consistency of the formulation or $R^*$ and is highly non-obvious in the traditional formulation.  In the Hopf-algebraic framework, the proof is straightforward.

\begin{proof}
    The proof works by using the projectors $\rho^{\text{UV}}$ and $\rho^{\text{IR}}$ together with the properties of the Hopf algebra of motic graphs:
\begin{gather*} (\id \otimes \Delta^{\text{IR}}_\mathcal{M})\circ \Delta^{\text{UV}}_\mathcal{M}= (\id \otimes [\id \otimes \rho^{\text{IR}}]\circ \Delta_\mathcal{M})\circ (\rho^{\text{UV}} \otimes \id) \circ \Delta_\mathcal{M}\\= (\rho^{\text{UV}} \otimes \id \otimes \rho^{\text{IR}}) \circ (\id\otimes \Delta_\mathcal{M})\circ \Delta_\mathcal{M}= (\rho^{\text{UV}} \otimes \id \otimes \rho^{\text{IR}}) \circ (\Delta_\mathcal{M} \otimes \id)\circ \Delta_\mathcal{M}\\= ([\rho^{\text{UV}} \otimes \id]\circ \Delta_\mathcal{M} \otimes \id)\circ(\id \otimes \rho^{\text{IR}}) \circ \Delta_\mathcal{M}= (\Delta^{\text{UV}}_\mathcal{M} \otimes \id)\circ\Delta^{\text{IR}}, \end{gather*}
where we used the coassociativity of the motic coproduct in the third step.
\end{proof}

It thus follows that the $R^*$-operator in the Hopf algebra formalism can be written in two ways. Either by treating the UV divergences first,
\begin{align} \begin{aligned} \label{eqn:RstarUV} R^*(\Gamma) &=m_\A^{(2)} \circ ( Z \otimes \phi \otimes \widetilde{Z}) \circ (\id \otimes \Delta^{IR}_\mathcal{M}) \circ\Delta^{UV}_\mathcal{M}(\Gamma) \\ &= \sum_{\gamma_{\text{UV}} \subseteq \Gamma}\sum_{\gamma_{mm} \subseteq \Gamma/\gamma_{\text{UV}}} Z(\gamma_{\text{UV}}) \hspace{1mm}\phi(\gamma_{\text{mm}}) \hspace{1mm} \widetilde{Z}(\Gamma /\gamma_{\text{UV}} / \gamma_{\text{mm}}), \end{aligned} \end{align}
or by handling the IR divergences first,
\begin{align} \begin{aligned} \label{eqn:RstarIR} R^*(\Gamma) &=m_\A^{(2)} \circ ( Z \otimes \phi \otimes \widetilde{Z}) \circ (\Delta^{UV}_\mathcal{M} \otimes \id) \circ \Delta^{IR}_\mathcal{M}(\Gamma) \\ &= \sum_{\gamma_{\text{mm}} \subseteq \Gamma}\sum_{\gamma_{UV} \subseteq \gamma_{\text{mm}}} Z(\gamma_{\text{UV}}) \hspace{1mm}\phi(\gamma_{\text{mm}}/\gamma_{\text{UV}}) \hspace{1mm} \widetilde{Z}(\Gamma / \gamma_{\text{mm}}). \end{aligned} \end{align}
The first version can be identified with the traditional definition in eq.~\eqref{eq:Rstar} after using the fact that the complimentary graph of the mass-momentum spanning subgraph can be identified with an IR-subgraph as proven in Theorem~\ref{thm:IRandMotic}. Note also that the condition that the UV- and IR-subgraph may not overlap, resolves itself nicely into the fact that the respective mass-momentum-spanning complementary graph shall be a subgraph of the contracted UV graph.
The second version provides an equivalent formulation of $R^*$. Note that the respective coaction ensures that $Z$ acts on 1PI superficially UV divergent subgraphs and $\widetilde{Z}$ acts on scaleless IR divergent graphs.

\subsection{\texorpdfstring{Examples of basic Hopf-algebraic $R^*$}{Examples of basic Hopf-algebraic R*}}
\label{sec:examples_hopf_basic}
To get acquainted with the Hopf-algebraic formulation, we will discuss a couple of concrete examples. 

Consider again the, in four spacetime dimensions, log-divergent graph 
\SetScale{0.15}
$\fullext.$
This graph contains both UV and IR (sub)divergences.  

The first step is to isolate the UV divergences, which we can do by acting with $\Delta^{UV}_\mathcal{M}$:
\begin{align} \label{eq:UVfirstExample} \Delta^{\text{UV}}_\mathcal{M}(\fullext)= (\rho^{\text{\text{UV}}} \otimes \id) \circ \Delta_\mathcal{M}(\fullext)=\extlegs23 \otimes \subgraphR14irdiv + \fullext \otimes \one + \one \otimes \fullext\,, \end{align}
where all subgraphs $\gamma$, appearing in the left entry of the tensor product, must be componentwise 1PI and satisfy $\omega(\delta)\geq 0$ for each 1PI component $\delta$. The empty graph also appears, as it is 1PI and superficially divergent by definition. 

To also include the IR divergences, we need to act with $(\id\otimes \Delta_\mathcal{M}^{IR}) \circ \Delta^{UV}_\mathcal{M}$ on the graph. This means that we have to apply $\Delta_\M^\text{IR}$ to the three different graphs that appear on the right hand side of the tensor product in eq.\ \eqref{eq:UVfirstExample}. The coaction $\Delta^{IR}_\mathcal{M}$ on these graphs yields
\begin{eqnarray}
\begin{aligned}
\label{eq:IRstructureExample}
\Delta^{\text{IR}}_\mathcal{M}(\one)&= \one \otimes \one \\
\Delta^{\text{IR}}_\mathcal{M}(\subgraphR14irdiv) &= \subgraphR14irdiv \otimes \one + \one \otimes \subgraphR14irdiv\\
\Delta^{\text{IR}}_\mathcal{M}(\fullext) &= \extlegs23 \otimes \subgraphR14irdiv + \fullext \otimes \one.
\end{aligned}
\end{eqnarray}\\
The empty graph does not appear as a subgraph in the last line because it is not mass-momentum-spanning. The subgraph $\fullext$ appears as it is trivially mass-momentum-spanning. The first two lines do contain the empty graph as a subgraph since every subgraph of a scaleless graph is mass-momentum-spanning. Recall the non-trivial inheritance of masses and momenta from parent to subgraphs in the context of motic (sub)graphs which was discussed in Section~\ref{setup:Motic}.

Combining the the coaction calculations in eqs.~\eqref{eq:UVfirstExample} and \eqref{eq:IRstructureExample}, we get
\begin{gather} \begin{gathered} \label{eq:completeUVIRexample} (\id \otimes \Delta^{\text{IR}}_\mathcal{M})\circ \Delta^{\text{UV}}_\mathcal{M}[\fullext] = (\id \otimes \Delta^{\text{IR}}_\mathcal{M})[\extlegs23 \otimes \subgraphR14irdiv + \fullext \otimes \one +\one \otimes\fullext]\\ =\extlegs23 \otimes \Delta^{\text{IR}}_\mathcal{M}(\subgraphR14irdiv) + \fullext \otimes \Delta^{\text{IR}}_\mathcal{M} (\one) +\one \otimes\Delta^{\text{IR}}_\mathcal{M}(\fullext)\\ =\extlegs23 \otimes (\subgraphR14irdiv \otimes \one + \one \otimes \subgraphR14irdiv) + \fullext \otimes (\one\otimes \one) + \one \otimes (\extlegs23 \otimes \subgraphR14irdiv + \fullext \otimes \one) \\ =\extlegs23 \otimes \subgraphR14irdiv \otimes \one+ \extlegs23 \otimes \one \otimes \subgraphR14irdiv+ \fullext \otimes \one \otimes \one +\one \otimes \extlegs23 \otimes \subgraphR14irdiv + \one \otimes \fullext \otimes \one. \end{gathered} \end{gather}
After acting with the operator $Z \otimes \phi \otimes \t Z$ on the triple tensor product, we obtain
\begin{gather*} (Z \otimes \phi \otimes \t Z) \circ (\id \otimes \Delta^{\text{\text{IR}}}_\mathcal{M})\circ \Delta^{\text{\text{UV}}}_\mathcal{M}[\fullext] = \\ Z(\extlegs23) \phi(\subgraphR14irdiv)\widetilde{Z}(\one) + Z(\extlegs23) \phi(\one) \widetilde{Z}(\subgraphR14irdiv) +\\ + Z(\fullext)\phi(\one)\widetilde{Z}(\one)+Z(\one) \phi(\extlegs23)\widetilde{Z}(\subgraphR14irdiv) + Z(\one)\phi(\fullext) \widetilde{Z}(\one), \end{gather*}
using $\phi(\one) = Z(\one)= \widetilde{Z}(\one) = 1$, we can see that this agrees with the calculation in eq.~(5.11)~in~\cite{Herzog:2017bjx}:
\vspace{-0.2cm}
\SetScale{0.2}
\begin{gather*} \label{Franzequation} R^*(\fullext)= \fullext + Z(\fullext) + Z(\extlegs23)* \externallegs14\\ + \widetilde{Z}(\subgraphR14irdiv)*\extlegs23 + \widetilde{Z}(\subgraphR14irdiv)*Z(\extlegs23). \end{gather*}
In contrast to the traditional formulation of $R^*$, the terms in the products appear in a different order in the Hopf algebra framework. The `bare' evaluation of the integral appears in between the UV counterterm on the left and the IR counterterm on the right. As the $*$-product boils down to ordinary multiplication for logarithmically divergent graphs, this is only an issue of non-logarithmically divergent graphs are involved. In this case the counterterm operator $\t Z$ involves differential operators, which act on the bare integral. These differential operators need to be left-acting in the Hopf algebra formulation.

To give another example of the workings of the Hopf algebra framework and to illustrate the coassociativity/commutativity of the IR and UV subtractions, we will give the evaluation of the same example, but we are going to start the calculation with the IR-coaction:
\vspace{-0.2cm}
\begin{align*} \Delta^{\text{IR}}_\mathcal{M}(\fullext)= \extlegs23 \otimes \subgraphR14irdiv + \fullext \otimes \one. \end{align*}
Performing the UV-coaction subsequently gives,
\begin{gather*} (\Delta^{\text{UV}}_\mathcal{M}\otimes \id) \circ \Delta^{\text{IR}}_\mathcal{M}(\fullext) \\ =(\Delta^{\text{UV}}_\mathcal{M}\otimes \id) \Big(\extlegs23 \otimes \subgraphR14irdiv + \fullext \otimes \one \Big)\\ =\extlegs23 \otimes \one \otimes \subgraphR14irdiv + \one \otimes \extlegs23 \otimes \subgraphR14irdiv + \extlegs23 \otimes \subgraphR14irdiv \otimes \one +\\ + \fullext \otimes \one \otimes \one + \one \otimes \fullext \otimes \one. \end{gather*}
The result agrees with the one in eq.\ \eqref{eq:completeUVIRexample} which confirms the statement of Theorem~\ref{thm:UVIRcommute} explicitly.

\subsection{Definitions of the IR and UV counterterms}
\label{sec:defIRUVhopf}

The $Z$- and $\t Z$-operators are \emph{twisted antipodes} \cite{Kreimer:1997dp} in the Hopf algebra language. In the literature, the twisted antipode for the UV counterterm is often denoted as $S^\phi_K$. This notation makes the implicit dependence of the counterterm operator on the Feynman rules, encoded in $\phi$, as well as on the renormalisation scheme, encoded in $K$, explicit. Furthermore, it highlights the similarity of the counterterm map with the antipode $S$ of the Hopf algebra, which will become apparent in this section. Here, we will stick to the notation of $Z$ and $\t Z$ for the twisted antipodes as the operators behave completely equivalently in the Hopf-algebraic formulation. The twisted antipodes $Z$ and $\t Z$ can be defined recursively in the Hopf algebra framework as,
\begin{align} \label{Zoperation} \ZZ(\Gamma)&= -K \circ m_\A^{(2)} \circ(\ZZ \otimes \phi \otimes \t \ZZ) \circ \left( \Delta_\mathcal{M}^{(2)} - \id \otimes \one \otimes \one \right) \circ \rho^{\text{UV}}(\Gamma)\,, \\ \label{Ztildeoperation} \t \ZZ(\Gamma)&= -K \circ m_\A^{(2)} \circ(\ZZ \otimes \phi \otimes \t \ZZ) \circ \left( \Delta_\mathcal{M}^{(2)} - \one \otimes \one \otimes \id \right) \circ \rho^{\text{IR}} (\Gamma)\,, \end{align}
where $m^{(2)}_\M= m_\M \circ (m_\M \otimes \id)=m_\M \circ (\id \otimes m_\M)$ is the multiplication operator acting on a triple tensor product, $\Delta^{(2)}_\M= (\Delta_\M \otimes \id)\circ \Delta_\M= (\id \otimes \Delta_\M)\circ \Delta_\M$ is the respective coproduct and the recursion terminates with $\ZZ(\one) = \t \ZZ(\one) = 1$.

As an example for a UV divergent graph $\Gamma$ we give
\begin{gather*} \ZZ(\Gamma)= -K \circ m_\A^{(2)} \circ(\ZZ \otimes \phi \otimes \t \ZZ) \circ \left( \Delta_\mathcal{M}^{(2)} - \id \otimes \one \otimes \one \right) \circ \rho^{\text{UV}} (\Gamma) \\ = -K \circ m_\A^{(2)} \circ( (\ZZ \circ \rho^{\text{UV}} )\otimes \phi \otimes (\t \ZZ \circ \rho^{\text{IR}})) \circ \left( \Delta_\mathcal{M}^{(2)} (\Gamma) - \Gamma \otimes \one \otimes \one \right) \\ = -K \circ m_\A^{(2)} \circ( \ZZ\otimes \phi \otimes \t \ZZ) \circ (\rho^\text{UV} \otimes \id \otimes \rho^\text{IR}) \circ \left( \Delta_\mathcal{M}^{(2)}( \Gamma) - \Gamma \otimes \one \otimes \one \right)\\ = -K \circ m_\A^{(2)} \circ( \ZZ\otimes \phi \otimes \t \ZZ) \circ \left((\rho^\text{UV} \otimes \id \otimes \rho^\text{IR}) \circ \Delta_\mathcal{M}^{(2)} (\Gamma) - \Gamma \otimes \one \otimes \one \right)\\ = -K \circ m_\A^{(2)} \circ( \ZZ\otimes \phi \otimes \t \ZZ) \circ \left(( \id \otimes \Delta_\mathcal{M}^{\text{IR}}) \circ \Delta_\mathcal{M}^{\text{UV}} (\Gamma) - \Gamma \otimes \one \otimes \one \right)\\ = -K \left( \sum_{\substack{\gamma_\text{UV} \subset \Gamma\\ \gamma_\text{UV} \neq \Gamma}} \sum_{\gamma_{mm} \subset \Gamma/\gamma_\text{UV}} \ZZ(\gamma_\text{UV}) \phi(\gamma_{mm}) \t \ZZ(\Gamma/\gamma_\text{UV}/\gamma_\text{mm}) \right), \end{gather*}
which is the same recursive equation as eq.~\eqref{eq:defZRstar} after we appropriately identify $\t \gamma$ and $\Gamma / \gamma_\text{UV} / \gamma_\text{mm}$ in accordance with the definition of $\t \gamma$. An analogous straightforward calculation results in the expression for $\t Z$ in eq.\ \eqref{eq:deftZRstar}. Note that the Hopf-algebraic version does not involve any choice on whether to consider the UV  or the IR divergences first and therefore provides a symmetric formulation of the $R^*$-operator.

This way we can give the most `unbiased' and compact formulation of the $R^*$-operator:
\begin{align} \label{eqn:short_R_star} R^* = \ZZ \star \phi \star \t \ZZ, \end{align}
where the $\star$-product\footnote{Unfortunately, there are two products that are completely different in nature, but appear in similar contexts in the mathematics and the physics literature on renormalisation: the $*$-operator from the original formulation of BPHZ is an `insertion' operator, which inserts information from a single Feynman diagram into another. The $\star$-product on the other hand is settled entirely in the Hopf algebra framework, where it is a normal group-product which multiplies linear maps from the Hopf algebra to some other space. Both notations are standard and are commonly used. So, we decided to keep the notation in line with the existing literature. The caveat is the problematic similarity between the two symbols $*$ and $\star$, which hopefully does not cause too much confusion for the reader.} is defined as
\begin{align*} \phi \star \psi := m_\A \circ ( \phi \otimes \psi ) \circ \Delta_\M. \end{align*}
It is easy to verify that eq.\ \eqref{eqn:short_R_star} is equivalent to eq.\ \eqref{eqn:RstarUV} and \eqref{eqn:RstarIR} by observing that
\begin{gather*} \ZZ \star \phi \star \t \ZZ = m_\A^{(2)} \circ ( \ZZ \otimes \phi \otimes \t \ZZ ) \circ \Delta_\mathcal{M}^{(2)} = m_\A^{(2)} \circ ( \ZZ \circ \rho^{\text{UV}}\otimes \phi \otimes \t \ZZ \circ \rho^{\text{IR}} ) \circ \Delta_\mathcal{M}^{(2)} \\ = m_\A^{(2)} \circ ( \ZZ\otimes \phi \otimes \t \ZZ) \circ ( \rho^{\text{UV}} \otimes \id \otimes \rho^{\text{IR}} ) \circ \Delta_\mathcal{M}^{(2)}, \end{gather*}
and using Theorem~\ref{thm:UVIRcommute}.

Due to coassociativity of $\Delta_\M$, this $\star$-product is \textit{associative}.  
Therefore the `commutativity' of UV  and IR subtraction is reduced to the following natural and obvious statement:
\begin{align} (\ZZ \star \phi) \star \t \ZZ = \ZZ \star (\phi \star \t \ZZ). \end{align}

Another immediate application of the Hopf algebra formulation of $R^*$ is a simple proof of the `finiteness' of $R^*$. Finiteness in this context means that for a graph $\Gamma$, $R^*(\Gamma)$ stays finite when we remove the regulator $\epsilon \rightarrow 0$.
\begin{theorem}
\label{thm:Rstarfinite}
If $\Gamma$ is a graph which is UV divergent or scaleless superficially IR divergent or both then,
\begin{align*} R^*( \Gamma ) \in \ker K. \end{align*}
\end{theorem}
Note that the space of finite amplitudes is equivalent to $\ker K$, as $K$ projects to the divergent part of the amplitude. 
\begin{proof}
Let 
\begin{align*} \bar R := m_\A^{(2)} \circ(\ZZ \otimes \phi \otimes \t \ZZ) \circ \left( \Delta_\mathcal{M}^{(2)} - \id \otimes \one \otimes \one - \one \otimes \one \otimes \id \right). \end{align*}
We have $Z = -K \circ ( \t Z + \bar R ) \circ \rho^{\text{UV}}$, $\t Z = -K \circ ( Z + \bar R ) \circ \rho^{\text{IR}}$ and 
$R^* = Z + \t Z + \bar R$.
Suppose that a graph $\Gamma$ is in $\ker \rho^{\text{IR}}$ but $\Gamma \notin \ker \rho^{\text{UV}}$.
Then,
\begin{align*} K \circ R^*(\Gamma) = K \circ \bar R(\Gamma) + K \circ Z(\Gamma) + K \circ \t Z(\Gamma) = K \circ \bar R(\Gamma) - K \circ ( \t Z + \bar R )(\Gamma) = 0, \end{align*}
and for  $\Gamma \notin \ker \rho^{\text{IR}}$ but $\Gamma \in \ker \rho^{\text{UV}}$ analogously.
For $\Gamma \notin \ker \rho^{\text{IR}}$ and $\Gamma \notin \ker \rho^{\text{UV}}$, we get
\begin{gather*} K \circ R^*(\Gamma) = K \circ \bar R(\Gamma) + K \circ Z(\Gamma) + K \circ \t Z(\Gamma) \\ = K \circ \bar R(\Gamma) - K \circ ( \t Z + \bar R )(\Gamma) + K \circ \t Z(\Gamma) = 0. \qedhere \end{gather*}
\end{proof}
Note the requirement that the graph has to be either superficially UV  or IR divergent. One of both is necessary without additional assumptions. See \cite{Borinsky:2015mga} for a discussion of the special role of \emph{overlapping-divergences} and the lattice structure of Feynman diagrams in this context.

\section{IR rearrangement}
\label{sec:IRrearrangement}

\subsection{A pedestrian example in detail}
As discussed in Section~\ref{setup:IRR} an IR rearrangement may be exploited to simplify the calculation of a UV counterterm. %
To illustrate the idea behind IR rearrangement and at the same time show how the Hopf algebra machinery can be applied, we will start with a basic example. Consider the graphs,
\vspace{-0.35cm}
\begin{align} \label{eqn:easy_example_graphs} \SetScale{0.5} \bubble,\qquad \vacuumoneloop\,. \end{align}
While both graphs are UV divergent only the second is IR divergent in four-dimensional spacetime as it has no legs. We can also think about the second graph, as being obtained from the first one after making both legs enter the same vertex. By Assumption~\ref{asp:firstIR}, both graphs give the same UV counterterm:
\vspace{-0.35cm}
\begin{align*} \SetScale{0.5} \ZZ \Big( \bubble \Big) = \ZZ \Big( \vacuumoneloop \Big) \end{align*}
We start with the first graph. The map $\ZZ$ can be evaluated in the Hopf algebra framework. For a start, we will follow the definitions very closely and do so in a meticulous manner. First, we compute the motic coproduct:
\vspace{-0.35cm}
\begin{align*} \SetScale{0.5} \Delta_\M \Big( \bubble \Big) &= \SetScale{0.5} \bubble \otimes \one + \one \otimes \bubble + \bubbleHalfTop \otimes \bubbleOneTop + \bubbleHalfBottom \otimes \bubbleOneBottom \\ &= \SetScale{0.5} \bubble \otimes \one + \one \otimes \bubble + 2 \bubbleHalfBottom \otimes \bubbleOneBottom      \end{align*}
as well as the double coproduct:
\vspace{-0.35cm}
\begin{gather} \begin{gathered} \label{eqn:bubble_dbl_cop} \SetScale{0.5} \Delta_\M^{(2)} \Big( \bubble \Big) = (\Delta_\M \otimes \id ) \circ \Delta_\M \Big( \bubble \Big) \\ = \SetScale{0.5} (\Delta_\M \otimes \id ) \Big( \bubble \otimes \one + \one \otimes \bubble + 2 \bubbleHalfBottom \otimes \bubbleOneBottom \Big)\\ = \SetScale{0.5} \bubble \otimes \one \otimes \one + \one \otimes \bubble \otimes \one+ 2 \bubbleHalfBottom \otimes \bubbleOneBottom \otimes \one + \\ + \SetScale{0.5} \one \otimes \one \otimes \bubble + 2 \bubbleHalfBottom \otimes \one \otimes \bubbleOneBottom + 2 \one \otimes \bubbleHalfBottom \otimes \bubbleOneBottom. \end{gathered} \end{gather}
With these calculations at hand we are ready to apply the definition of the $\ZZ$-operator to evaluate the respective counterterms:
\vspace{-0.35cm}
\begin{gather} \begin{gathered} \label{eqn:Sbubble} \SetScale{0.5} \ZZ \Big( \bubble \Big) = - K \circ m_A^{(2)} \circ(\ZZ \otimes \phi \otimes \t \ZZ) \circ \left( \Delta_\mathcal{M}^{(2)} - \id \otimes \one \otimes \one \right) \circ \rho^{\text{UV}}\Big( \bubble \Big) \\ \SetScale{0.5} = - K \circ m_A^{(2)} \circ(\ZZ \otimes \phi \otimes \t \ZZ) \left( \Delta_\mathcal{M}^{(2)}\Big(\bubble\Big) - \bubble \otimes \one \otimes \one \right) \\ \SetScale{0.5} = - K \circ m_A^{(2)} \circ(\ZZ \otimes \phi \otimes \t \ZZ) \Big( \one \otimes \bubble \otimes \one+ 2 \bubbleHalfBottom \otimes \bubbleOneBottom \otimes \one + \\ + \SetScale{0.5} \one \otimes \one \otimes \bubble + 2 \bubbleHalfBottom \otimes \one \otimes \bubbleOneBottom + 2 \one \otimes \bubbleHalfBottom \otimes \bubbleOneBottom \Big) \end{gathered} \end{gather}
As $\ZZ = \ZZ \circ \rho^{\text{UV}}$ and $\t \ZZ = \t \ZZ \circ \rho^{\text{IR}}$, all terms vanish which either do not involve a UV divergent graph in the first entry of the tensor product or do not have a scaleless IR divergent graph in the third entry of the tensor product.
In four dimensions we have, 
\vspace{-0.35cm}
\begin{gather*} \SetScale{0.5} \omega \Big( \bubble \Big) = \omega \Big( \vacuumoneloop \Big) = 0,\qquad \SetScale{0.5} \omega \Big(\bubbleOneBottom \Big) = 2. \end{gather*}
Therefore, both $\SetScale{0.3} \bubble$ and $\SetScale{0.3} \vacuumoneloop$ are logarithmically divergent. The motic subgraph $\SetScale{0.3} \bubbleHalfBottom$ is not a 1PI graph. The 1PI graph $\SetScale{0.3} \bubble$ is not scaleless. From this we can deduce that 
\vspace{-0.35cm}
\begin{gather*} \SetScale{0.5} \rho^{\text{UV}}\Big(\bubbleHalfBottom\Big) = 0\,,\qquad \SetScale{0.5} \rho^{\text{IR}}\Big(\bubble\Big) = 0 = \rho^{\text{IR}}\Big(\bubbleOneBottom\Big). \end{gather*}

The respective projectors act as the identity map on all other graphs involved.

 Therefore, the expression in eq.\ \eqref{eqn:Sbubble} simplifies significantly. In fact, only the first triple tensor product term in the last equality of eq.\ \eqref{eqn:Sbubble} contributes:
 \vspace{-0.35cm}
\begin{gather*} \SetScale{0.5} \ZZ \Big( \bubble \Big) = - K \circ m_A^{(2)} \circ(\ZZ \otimes \phi \otimes \t \ZZ) \Big( \one \otimes \bubble \otimes \one \Big) \SetScale{0.5} = - K \left[ \phi\Big(\bubble\Big) \right] \end{gather*}
This finishes the combinatorial part of the calculation as we completely reduced the counterterm $\SetScale{0.3} \ZZ(\bubble)$ to evaluations of the Feynman rules $\phi$. In this simple example, the evaluation of $\SetScale{0.3} \phi(\bubble)$ is very easy and an IR rearrangement results in no computational advantage. However, it is still instructive to evaluate the counterterm in another way. 

We will continue with the second graph in eq.\ \eqref{eqn:easy_example_graphs} and repeat the calculation. Evaluating the double coproduct is straightforward:
\vspace{-0.35cm}
\begin{gather*} \SetScale{0.5} \Delta_\M^{(2)} \Big( \vacuumoneloop \Big) = \vacuumoneloop \otimes\one \otimes \one + \one \otimes \vacuumoneloop \otimes \one + \one \otimes \one \otimes \vacuumoneloop \end{gather*}
As $ \SetScale{0.3} \vacuumoneloop$ does not have any non-trivial motic subgraphs.
Therefore,
\vspace{-0.35cm}
\begin{gather*} \SetScale{0.5} \ZZ \Big( \vacuumoneloop \Big) = - K \circ m_A^{(2)} \circ(\ZZ \otimes \phi \otimes \t \ZZ) \Big( \one \otimes \vacuumoneloop \otimes \one + \one \otimes \one \otimes \vacuumoneloop \Big) \\ \SetScale{0.5} = - K \left[ \phi\Big(\vacuumoneloop\Big) + \t \ZZ \Big( \vacuumoneloop \Big) \right] = - K \left[ \phi\Big(\vacuumoneloop\Big) \right] - \t \ZZ \Big( \vacuumoneloop \Big), \end{gather*}
where we used that $K$ is a projector in the last step.

We obtain a relation between $\SetScale{0.3} \ZZ ( \vacuumoneloop )$ and $\SetScale{0.3} \t \ZZ ( \vacuumoneloop )$. Using the vanishing of scaleless integrals under $\phi$ we then get the simple relation:
\vspace{-0.35cm}
\begin{align*} \SetScale{0.5} \ZZ \Big( \vacuumoneloop \Big) = - \t \ZZ \Big( \vacuumoneloop \Big), \end{align*}
which in fact holds generally for all $1$-loop scaleless logarithmically divergent Feynman graphs. In Section~\ref{sec:antipode}, we will generalise this observation and produce a simple algebraic relation between the UV  and IR counterterms of any scaleless logarithmically divergent Feynman graph.

Putting both ways to calculate the UV counterterm together, we therefore get an expression for both, the UV  and the IR counterterms for all involved graphs
\vspace{-0.35cm}
\begin{align*} \SetScale{0.5} \ZZ \Big( \vacuumoneloop \Big) &= \SetScale{0.5} \ZZ \Big( \bubble \Big)= - K \left[ \phi\Big(\bubble\Big) \right]\,, \\ \SetScale{0.5} \t \ZZ \Big( \vacuumoneloop \Big) &= \SetScale{0.5} K \left[ \phi\Big(\bubble\Big) \right] \,,\\ \SetScale{0.5} \t \ZZ \Big( \bubble \Big) &=0 \end{align*}
The last line follows from the fact that the graph in the argument is not scaleless. In more complicated examples, we will get more complicated relations between the UV  and IR counterterms. In this simple example the utility of these relations does not become fully apparent. The calculation of the value $\SetScale{0.3} K \left[ \phi\Big(\bubble\Big) \right] = \frac{1}{\epsilon}$ in dimensional regularisation with minimal subtraction is already the simplest way to calculate the respective counterterms. To illustrate the actual utility, we will discuss a more involved example.

\clearpage
\subsection{A three-loop example}
As a non-trivial example we consider two different three-loop Feynman diagrams,
\vspace{-0.35cm}
\begin{align} \label{eqn:hard_example} \SetScale{0.5} \diagramthreeloopA\,,\qquad \diagramthreeloopB \,. \end{align}
These graphs are logarithmically divergent in four-dimensional spacetime and differ only by the positions of their external legs. Therefore, their images under the $\ZZ$-operator are equal by Assumption~\ref{asp:firstIR}. In practice, there is an important difference between the two diagrams: the evaluation of the Feynman integral for the first diagram is slightly harder than the evaluation of the integral for the second one because of the integration trick explained in Section~\ref{setup:IRR}. Therefore, it is preferred to evaluate $\SetScale{0.3} K[ \phi (\diagramthreeloopB)]$ instead of $\SetScale{0.3} K[ \phi (\diagramthreeloopA)]$ in the course of calculating the respective UV counterterms. However, the choice of the graph $\SetScale{0.3} \diagramthreeloopB$ for the evaluation of the counterterm comes with a caveat: by `rerouting' the external momenta, new IR divergences are introduced into the graph, which were not present in $\SetScale{0.3} \diagramthreeloopA$, which is free of IR divergences. The sole practical purpose of the $R^*$-procedure in this context is therefore to isolate these newly introduced IR divergences and account for them in a consistent manner. The counterterm of both diagrams can then be obtained by evaluating the simpler Feynman integral for the graph $\SetScale{0.3} K[ \phi (\diagramthreeloopB)]$ and other Feynman integrals that are strictly less complicated. In the following example this procedure is illustrated in detail. 

We start by deriving an expression for the UV counterterm $\ZZ ( \SetScale{0.3} \diagramthreeloopA)$. Even though the formulas are more compact while working with the full motic coaction $\Delta_\M$, it is much more efficient to use the dedicated UV- and IR-coactions. The reason for this is the quick proliferation of terms in the triple coproducts, which are eventually projected to a small number of terms. In the following we will make use of the dedicated coactions and project out motic subgraphs, which do not lead to either a UV  or IR divergence, as soon as possible. 

We start by applying $\Delta^\text{IR}_\M$,
\vspace{-0.35cm}
\begin{gather} \label{eqn:simpleIRcoaction} \SetScale{0.5} \Delta_\M^\text{IR} \Big( \diagramthreeloopA \Big) = \diagramthreeloopA \otimes \one. \end{gather}
Of course, this simple result stems from the fact that $ \SetScale{0.3} \diagramthreeloopA$ does not have any (sub) IR divergences.
However, it has UV divergences, which can be captured by a subsequent application of the $\Delta^{\text{UV}}_\M$ coaction, %
\vspace{-0.35cm}
\begin{gather*} \SetScale{0.5} (\Delta^{\text{UV}}_\mathcal{M} \otimes \id) \circ \Delta_\M^\text{IR} (\diagramthreeloopA)= \\ \SetScale{0.5} \Big( \diagramthreeloopA \otimes \one + \vacuumoneloop \otimes \twoloopunderA + \vacuumoneloop \otimes \diagram2loopA\\ + \SetScale{0.5}\one \otimes \diagramthreeloopA + \vacuumoneloop \vacuumoneloop \otimes \bubble \Big) \otimes \one . \end{gather*}
Notice that the second and third term are equal. Hence, they will be combined. In the next step we act with the counterterm operators on the triple tensor product and multiply the result
\vspace{-0.35cm}
\begin{align*} &\SetScale{0.5}\ZZ(\diagramthreeloopA)= -K \circ m_A^{(2)} \circ(\ZZ \otimes \phi \otimes \t \ZZ) \circ \left( \Delta_\mathcal{M}^{(2)} - \id \otimes \one \otimes \one \right) \circ \rho^{\text{UV}} \SetScale{0.5} \Big(\diagramthreeloopA \Big) \\ &= -K \circ m_A^{(2)} \circ(\ZZ \otimes \phi \otimes \t \ZZ) \left( (\Delta^{UV}_\mathcal{M} \otimes \id) \circ \Delta_\M^\text{IR} - \id \otimes \one \otimes \one \right) \circ \rho^{\text{UV}} \SetScale{0.5} \Big(\diagramthreeloopA \Big) \\ &= \SetScale{0.5} -K \circ m_A^{(2)} \circ(\ZZ \otimes \phi \otimes \t \ZZ) \left( (\Delta^{UV}_\mathcal{M} \otimes \id) \circ \Delta_\M^\text{IR}\Big(\diagramthreeloopA\Big) - \diagramthreeloopA \otimes \one \otimes \one \right) \\ &= \SetScale{0.5} -K \circ m_A^{(2)} \circ(\ZZ \otimes \phi \otimes \t \ZZ) \Big( 2 \vacuumoneloop \otimes \twoloopunderA \otimes \one+ \SetScale{0.5}\one \otimes \diagramthreeloopA \otimes \one \\ &\qquad \SetScale{0.5} \qquad\qquad \qquad\qquad \qquad\qquad + \vacuumoneloop \vacuumoneloop \otimes \bubble \otimes \one \Big) \\ &= \SetScale{0.5}-K\left[\phi\Big(\diagramthreeloopA \Big) + 2 \ZZ\Big(\vacuumoneloop\Big) \phi\Big(\twoloopunderA \Big) + \Big( \ZZ \Big(\vacuumoneloop\Big) \Big)^2 \phi\Big(\bubble\Big)\right]. \end{align*}
As expected, we do not need the $\t \ZZ$-operator in this case, as there are no IR divergences. 
To evaluate the counterterm, we still have to apply the $\phi$ map to the Feynman graphs $\SetScale{0.3} \diagramthreeloopA$ and $\SetScale{0.3} \twoloopunderA$ and use our results from the last example.
The respective Feynman integrals can be evaluated using the FORM \cite{Ruijl:2017dtg} program FORCER \cite{tuLL2016} to give
\begin{gather*} \SetScale{0.5} \ZZ\Big(\diagramthreeloopA\Big)= -K\Big[\Big(-\frac{1}{\epsilon^3}-\frac{5}{\epsilon^2}-\frac{19}{\epsilon}\Big)+\Big(\frac{1}{3\epsilon^3}+\frac{8}{3\epsilon^2}+\frac{44}{3\epsilon}\Big)+\Big(\frac{1}{\epsilon^3}+\frac{2}{\epsilon^2}+\frac{4}{\epsilon}\Big) + \O(\epsilon^0)\Big]\\ = -\frac{1}{3\epsilon^3}+\frac{1}{3\epsilon^2}+\frac{1}{3\epsilon}. \end{gather*}
Let us now turn to the calculation of the second graph in eq.\ \eqref{eqn:hard_example}. We start by applying the IR-coaction,
\vspace{-0.35cm}
\begin{gather*} \SetScale{0.5} \Delta_\mathcal{M}^\text{IR}(\diagramthreeloopB)=\diagramthreeloopB \otimes \one + \bubble \otimes \diagramtwovacuum + \bubble \vacuumoneloop \otimes \vacuumoneloop\,. \end{gather*}
Note that, due to the presence of IR divergences produced during the IR rearrangement, the expression is significantly more involved than eq.~\eqref{eqn:simpleIRcoaction}.
We continue by acting with the double coaction $(\Delta_\M^\text{UV} \otimes \id) \circ \Delta_\M^\text{IR}$ to get
\vspace{-0.35cm}
\begin{gather*} \SetScale{0.5} \SetScale{0.5}(\Delta_\M^\text{UV} \otimes \id)\circ \Delta_\M^\text{IR}\Big(\diagramthreeloopB\Big) = \\ \SetScale{0.5}\Big(\diagramthreeloopB \otimes \one + \bubble \otimes \diagramtwovacuum + \vacuumoneloop \otimes \diagramtwoloopB + \one \otimes \diagramthreeloopB + \\ \SetScale{0.5} + \bubble \vacuumoneloop \otimes \vacuumoneloop \Big) \otimes \one +\Big(\bubble \otimes \one + \one \otimes \bubble \Big) \otimes \diagramtwovacuum \\ \SetScale{0.5}+\Big(\bubble \vacuumoneloop \otimes \one + \bubble \otimes \vacuumoneloop +\vacuumoneloop \otimes \bubble \\ \SetScale{0.5}+\one \otimes \bubble \vacuumoneloop\Big) \otimes \vacuumoneloop. \end{gather*}
We can now apply the counterterm operators and the Feynman rules in the form of the triple tensor product map $K \circ m_A^{(2)} \circ(\ZZ \otimes \phi \otimes \t \ZZ)$. Taking the vanishing of scaleless graphs under $\phi$ in the second entry of the tensor product into account, we find:
\vspace{-0.35cm}
\begin{gather*} \SetScale{0.5} \ZZ \Big(\diagramthreeloopB\Big) = \\ \SetScale{0.5}-K \circ m_A^{(2)} \circ(\ZZ \otimes \phi \otimes \t \ZZ) \circ \Big((\Delta_\M^\text{UV} \otimes \id) \Delta_\M^\text{IR}\Big(\diagramthreeloopB\Big) - \Big(\diagramthreeloopB\Big) \otimes \one \otimes \one \Big)= \\ \SetScale{0.5} -K\Big[\ZZ\Big(\vacuumoneloop\Big)\phi\Big(\diagramtwoloopB\Big) + \phi\Big(\diagramthreeloopB\Big) \SetScale{0.5}+ \ZZ\Big(\bubble\Big) \t \ZZ\Big(\diagramtwovacuum\Big) \\ +\SetScale{0.5} \phi \Big(\bubble \Big) \t \ZZ \Big(\diagramtwovacuum\Big) +\ZZ\Big(\bubble\Big) \ZZ\Big(\vacuumoneloop\Big) \t \ZZ\Big( \vacuumoneloop\Big) \\ +\SetScale{0.5} \ZZ \Big(\vacuumoneloop\Big) \phi \Big( \bubble \Big) \t \ZZ \Big( \vacuumoneloop \Big)\Big]. \end{gather*}
Using the identities from the one-loop example we get,
\vspace{-0.35cm}
\begin{gather*} \SetScale{0.5} \ZZ \Big(\diagramthreeloopB\Big) = -K\Big[\phi\Big(\diagramthreeloopB\Big)-\frac{1}{\epsilon}\phi\Big(\diagramtwoloopB\Big) \\ \SetScale{0.5} +\Big( \phi \Big(\bubble \Big)- \frac{1}{\epsilon} \Big) \t \ZZ\Big(\diagramtwovacuum\Big) + \frac{1}{\epsilon^3} - \frac{1}{\epsilon^2} \phi \Big( \bubble \Big) \Big] \end{gather*}
The only remaining IR counterterm is $\SetScale{0.3} \t \ZZ(\diagramtwovacuum)$, which we eventually also want to express using UV counterterms which we can evaluate explicitly.
We have,
\vspace{-0.35cm}
\begin{gather*} \SetScale{0.5}(\Delta_\M^\text{UV} \otimes \id)\circ \Delta_\M^\text{IR}\Big(\diagramtwovacuum\Big) = \\ \SetScale{0.5} (\Delta_\M^\text{UV} \otimes \id) \Big( \diagramtwovacuum \otimes \one + \vacuumoneloop \otimes \vacuumoneloop + \one \otimes \diagramtwovacuum \Big) \\ =\SetScale{0.5} \diagramtwovacuum \otimes \one \otimes \one + \vacuumoneloop \otimes \vacuumoneloop \otimes \one + \one \otimes \diagramtwovacuum \otimes \one \\ \SetScale{0.5} + \vacuumoneloop \otimes \one \otimes \vacuumoneloop + \one \otimes \vacuumoneloop \otimes \vacuumoneloop + \one \otimes \one \otimes \diagramtwovacuum \end{gather*}
Only the first, third and sixth term in the coaction calculation survive the application of the $\ZZ \otimes \phi \otimes \t \ZZ$ operator, because the factors in the middle of the triple coproduct are scaleless and $\phi$ maps them to $0$. Therefore, 
\vspace{-0.35cm}
\begin{gather} \begin{gathered} \label{eq:tSdotbanana} \SetScale{0.5} \t \ZZ \Big(\diagramtwovacuum\Big) = -K \circ m_A^{(2)} \circ(\ZZ \otimes \phi \otimes \t \ZZ) \Big( (\Delta_\M^\text{UV} \otimes \id)\circ \Delta_\M^\text{IR}\Big(\diagramtwovacuum\Big) - \one \otimes \one \otimes \diagramtwovacuum \Big) \\ \SetScale{0.5} = -K \circ m_A^{(2)} \circ(\ZZ \otimes \phi \otimes \t \ZZ) \Big( \diagramtwovacuum \otimes \one \otimes \one + \vacuumoneloop \otimes \one \otimes \vacuumoneloop \Big) \\ = \SetScale{0.5} -K\Big[ \ZZ \Big(\diagramtwovacuum \Big) + \ZZ \Big(\vacuumoneloop \Big) \t \ZZ \Big(\vacuumoneloop \Big) \Big] \\ \SetScale{0.5} = - \ZZ \Big(\diagramtwovacuum \Big) + \Big(\ZZ \Big(\vacuumoneloop \Big) \Big)^2 = - \ZZ \Big(\diagramtwovacuum \Big) + \frac{1}{\epsilon^2}, \end{gathered} \end{gather}
where we used the calculations of the previous example in the last two equalities.

The last task is to evaluate the remaining UV counterterm $ \SetScale{0.3} \ZZ \Big(\diagramtwovacuum \Big)$. For this we will again use IR rearrangement:
\vspace{-0.35cm}
\begin{gather*} \SetScale{0.5} \ZZ \Big(\diagramtwovacuum \Big) = \ZZ \Big(\diagram2loopA \Big) \end{gather*}
By adding external lines, we are dealing with an IR finite Feynman diagram, which can be evaluated via the Feynman rules.
The double coaction gives,
\vspace{-0.35cm}
\begin{gather*} \SetScale{0.5}(\Delta_\M^\text{UV} \otimes \id)\circ \Delta_\M^\text{IR}\Big(\diagram2loopA\Big) = \diagram2loopA \otimes \one \otimes \one + \vacuumoneloop \otimes \bubble \otimes \one + \one \otimes \diagram2loopA\otimes \one. \end{gather*}
As expected, there are no new IR divergences on the third factor of the triple tensor products of the double coaction.
This results in the UV counterterm,
\vspace{-0.35cm}
\begin{gather*} \SetScale{0.5} \ZZ \Big(\diagram2loopA \Big) = -K \Big[ \phi\Big(\diagram2loopA \Big) + \ZZ\Big(\vacuumoneloop\Big) \phi \Big(\bubble \Big)\Big] \\ \SetScale{0.5} = -K \Big[ \phi\Big(\diagram2loopA \Big) \Big] + \frac{1}{\epsilon^2} \end{gather*}
and for the IR counterterm,
$\SetScale{0.3} \t \ZZ (\diagramtwovacuum ) = K [ \phi(\diagram2loopA ) ]$. Note that we could have also used the simpler graph 
$\SetScale{0.3} \diagramtwoloopBt$ for the calculation of those counterterms. The price for this would have been slightly more involved combinatorics for the evaluation of the respective terms.

Substituting this result into the expression for the counterterm $\SetScale{0.3} \ZZ (\diagramthreeloopB)$ gives,
\vspace{-0.35cm}
\begin{gather*} \SetScale{0.5} \ZZ \Big(\diagramthreeloopB\Big) = \SetScale{0.5} -K\Big[\phi\Big(\diagramthreeloopB\Big)-\frac{1}{\epsilon}\phi\Big(\diagramtwoloopB\Big) \\ \SetScale{0.5} +\Big( \phi \Big(\bubble \Big)- \frac{1}{\epsilon} \Big) K \Big[ \phi(\diagram2loopA ) \Big] + \frac{1}{\epsilon^3} - \frac{1}{\epsilon^2} \phi \Big( \bubble \Big) \Big]. \end{gather*}
Observe that we only have simpler graphs than the original graph $\SetScale{0.3} \diagramthreeloopA$ which have to be evaluated using $\phi$ on the right hand side. The combinatorics of the calculation using IR rearrangement becomes much more involved than the combinatorics for the original calculation which did not involve any IR divergences. The advantage is the simpler structure of the Feynman graphs that have to be evaluated.

We may again evaluate the now simpler graphs on the right hand side to obtain
\vspace{-0.35cm}
\begin{gather*} \SetScale{0.5}\ZZ\Big(\diagramthreeloopB\Big)= -K\Big[ \Big(-\frac{1}{6\epsilon^3}-\frac{5}{6\epsilon^2}-\frac{23}{6\epsilon}\Big) +\frac{1}{\epsilon} \Big(\frac{1}{2\epsilon^2}+\frac{3}{2\epsilon}+\frac{9}{2}\Big) -\frac{1}{\epsilon^2} -\frac{1}{\epsilon} + \O(\epsilon^0) \Big] \\ =-\frac{1}{3\epsilon^3}+\frac{1}{3\epsilon^2} +\frac{1}{3\epsilon}, \end{gather*}
which agrees with the counterterm $\SetScale{0.3} \ZZ(\diagramthreeloopA)$ as expected. 
\section{An antipodal relation between UV and IR}
\label{sec:antipode}
The close relationship between the UV  and IR counterterms can be made more explicit in the Hopf algebra language. Especially, if we restrict ourselves to exclusively scaleless and logarithmically divergent graphs. Under this restriction we can express $\t \ZZ$-operator in terms of the $\ZZ$-operator. The precise relation is formulated in the following.
\begin{theorem}
If $\Gamma$ is a scaleless log-divergent graph with only log-subdivergences, then
\begin{align} \t \ZZ(\Gamma) &= \ZZ \circ S(\Gamma) & \ZZ(\Gamma) &= \t \ZZ \circ S(\Gamma), \end{align}
where $S$ is the antipode of the Hopf algebra of motic graphs.
\end{theorem}
\begin{proof}
Because $\Gamma$ is log-divergent and scaleless, $\rho^\text{IR}(\Gamma) = \Gamma$, and we may apply Theorem~\ref{thm:Rstarfinite}. It follows that $R^*(\Gamma) \in \ker K$. As $\phi(\Lambda) =0$ for all non-trivial scaleless graphs $\Lambda$ and because a scaleless graph has only scaleless subgraphs, we also have $R^*(\Gamma) = Z \star \phi \star \t Z = (Z \star \unit_\A \circ \counit_\M \star \t Z)(\Gamma) = (Z \star \t Z)(\Gamma) \in \text{im} K$. Because $K^2 = K$, we get $\ker K \cap \text{im} K=0$ and $R^*(\Gamma) = 0$. Therefore we have, $0 = (Z \star \t Z)(\Gamma)$, which means that $Z$ and $\t Z$ are mutually inverse under the convolution product on exclusively log-divergent graphs. The statement follows.
\end{proof}
A similar theorem holds if we lift the log-divergence assumption. However, we would need to deal with derivative and polynomial operators in the external momenta and such an analysis lies beyond the scope of the present paper. 
Note that this fixes the IR counterterm operator to be the exact dual of the UV counterterm operator under the antipode involution.

As an example of this, in explicit calculations also very useful relation, we start with the simple scaleless Feynman diagram $\SetScale{0.3} \vacuumoneloop$. Its motic antipode is simply given by
\begin{align*} \SetScale{0.5} S\Big( \vacuumoneloop \Big) = - \vacuumoneloop. \end{align*}
Therefore we recover the simple relation
\begin{align*} \SetScale{0.5} \t \ZZ\Big( \vacuumoneloop \Big) = \ZZ \circ S\Big( \vacuumoneloop \Big) = - \ZZ \Big( \vacuumoneloop \Big), \end{align*}
which we already exploited in the last section.

As a slightly more complicated example, consider the graph $\SetScale{0.2}\graphR$, which is also log-divergent.
Its motic coproduct is,
\begin{align*} \Delta_\M\Big(\SetScale{0.2}\graphR\Big) = \SetScale{0.2}\graphR \otimes \one +\SetScale{0.5}\vacuumoneloop \otimes \vacuumoneloop +\SetScale{0.2} \michitriangle \otimes \michitadpole +\SetScale{0.2} \one \otimes \graphR. \end{align*}
Applying the defining identity of the antipode $\unit_\M \circ \counit_\M = \id \star S = S \star \id$ gives,
\begin{gather*}    0 = \unit_\M \circ \counit_\M\Big( \SetScale{0.2}\graphR\Big) = (S \star \id ) \Big(\SetScale{0.2}\graphR\Big) = m_\M \circ ( S \otimes \id ) \circ \Delta_\M \Big(\SetScale{0.2}\graphR\Big) \\ = S\Big(\SetScale{0.2}\graphR\Big) +S\Big(\SetScale{0.5}\vacuumoneloop\Big)\vacuumoneloop+ S\Big(\SetScale{0.2} \michitriangle \Big) \michitadpole + S(\one ) \graphR . \end{gather*}
Also using the results of the trivial one-loop antipode calculations,
\begin{align*} S\Big(\SetScale{0.5}\vacuumoneloop\Big) &= -\SetScale{0.5}\vacuumoneloop \\ S\Big(\SetScale{0.2} \michitriangle \Big) &= - \michitriangle, \end{align*}
as well as $S(\one) = \one$ results in
\begin{align*} S\Big(\SetScale{0.2}\graphR\Big) &= - \graphR +\Big(\SetScale{0.5}\vacuumoneloop\Big) ^2 + \SetScale{0.2} \michitriangle \michitadpole \end{align*}
It follows that 
\begin{align*} \t \ZZ\Big(\SetScale{0.2}\graphR\Big) = \ZZ \circ S \Big( \SetScale{0.2}\graphR\Big) = - \ZZ \Big(\graphR \Big) +\ZZ \Big(\SetScale{0.5}\vacuumoneloop\Big) ^2 + \ZZ \Big(\SetScale{0.2} \michitriangle \Big)\ZZ \Big( \michitadpole \Big). \end{align*}
The last term vanishes because $\ZZ (\SetScale{0.15} \michitriangle ) = 0$ as the graph is not UV divergent. Therefore, we get the result which was expected from the direct calculation in eq.\ \eqref{eq:tSdotbanana}.
\section{Discussion and Outlook}
\label{sec:conclusion}

We have established a new Hopf-algebraic formulation of the $R^*$-operation for log-divergent graphs. 
In particular we proved a theorem on the equivalence of infrared subgraphs in the original definition of the $R^*$-operation 
and mass-momentum spanning motic subgraphs, which were introduced by Brown. 
Based on Brown's motic Hopf algebra we introduced two coactions which respectively project out IR  and UV divergent subgraphs and which allow to 
organise the combinatorics of the $R^*$-operation. Using the coassociativity of the 
coproduct of the motic Hopf algebra we proved a coassociativity relation of the IR- and UV-coactions. From this it also follows that the subtraction of IR and UV divergences `commutes'. To illustrate the new formulation we provided numerous examples with explicit calculations. 

Finally we discussed the structure of the IR  and UV  counterterms and showed that they can be related to each other using the antipode of the motic Hopf algebra. This last property in particular shows that the IR and UV counterterm operators are inverse under the Hopf algebra convolution product, thereby providing a clear algebraic picture of the duality between IR and UV divergences. 

There exist many possible directions for future extensions of this research. For once it would be interesting to extend the Hopf-algebraic formulation beyond the log-divergent case. There certainly exist obstacles for accomplishing this since the graphs must be endowed with tensor structure, which requires graphs in different tensor entries to have knowledge of each other which in turn creates difficulties for keeping the Hopf algebra structure. A possible first step to solve this problem might be a bialgebra construction due to Kock \cite{Kock:2015bya}.

With the mathematically rigid formulation of the IR divergences in the Hopf algebra framework, the door is now also open for future investigations along the lines of \cite{Kreimer:2005rw,Kreimer:2006ua,vanSuijlekom:2006fk} on the mathematical structure of IR singularities in the context of perturbative quantum field theories. Our algebraic formulation also immediately suggests itself to be implemented in computer algebra as it was done in \cite{Borinsky:2014xwa} for the Connes-Kreimer Hopf algebra.

Although we achieved a rigid definition of the $R^*$-operation and the associated counterterm operators in the Hopf algebra language, there are many formal questions that remain open. Obviously, Assumption~\ref{asp:firstIR} demands for a rigorous proof. Moreover, our analysis suggests that the statement of Assumption~\ref{asp:firstIR}, together with additional conditions and the finiteness of the $R^*$-operation from Theorem~\ref{thm:Rstarfinite} might serve as a complete \emph{definition} of the IR counterterm operator $\t Z$. An interpretation of this counterterm operator in terms of a generalised Riemann-Hilbert problem as in \cite{Connes:1999yr} might also be feasible.

Finally, a striking question  emerges from our result: can the algebraic structure of the Euclidean IR renormalisation procedure also be extended to the Minkowskian case? Already much progress has been made towards establishing a forest-like subtraction formula for general IR divergences of both soft and collinear type for both loop corrections \cite{Humpert:1981jw,Collins:1981uk,Collins:2011zzd,Erdogan:2014gha} as well as real radiation corrections; see for instance \cite{Caola:2017dug,Magnea:2018ebr,Herzog:2018ily} for some recent developments in this direction at fixed order in perturbation theory. In particular the recent work \cite{Ma:2019hjq}, where a forest formula for wide-angle scattering is derived based on an algebra of normal spaces, could prove useful for an extended algebraic formulation.

\section*{Acknowledgements}
This work has been supported by the the NWO Vidi grant 680-47-551  ``Decoding Singularities of Feynman graphs'' and the UKRI Future Leaders Fellowship ``Forest Formulas for the LHC'' Mr/S03479x/1. We are grateful to F.\ Brown and D.\ Kreimer for useful discussions and comments on the manuscript. We would also like to thank D. Broadhurst and J. Roosmale for comments. Furthermore, the authors thank all the members of the inner circle.

\providecommand{\href}[2]{#2}\begingroup\raggedright\endgroup

\end{document}